%% file: paper.tex
\documentclass[a4paper,11pt]{article}
\pdfoutput=1
\usepackage[margin=30mm]{geometry}

\usepackage[utf8]{inputenc}
\usepackage{newunicodechar}
\usepackage{sty/unicodeconfig}
\usepackage[english]{babel}
\usepackage{microtype}

\usepackage[usenames,dvipsnames]{xcolor}
\usepackage{sty/colors}
\usepackage[fleqn]{amsmath}
\usepackage{amssymb}
\usepackage{mathtools}
\usepackage{sty/widebar}
\usepackage{centernot}
\usepackage{stmaryrd}
\usepackage{pifont}
\usepackage{sty/mathpartir}

\usepackage{amsthm}
\usepackage{thmtools}
\usepackage[framemethod=tikz]{mdframed}
\usepackage{tablefootnote}
\usepackage{sty/theoremconfig}

\usepackage{tikz}
\usepackage{sty/tikzconfig}
\usetikzlibrary{shapes.geometric,shapes.symbols}
\input{fig/graph_templates.tex}

\usepackage{sty/colors}

\usepackage{caption}
\usepackage{booktabs}
\usepackage{scalefnt}

\usepackage{enumitem}
\usepackage{xifthen}
\usepackage{hyperref}
\hypersetup{pdftitle={Distributed Graph Automata and Verification of Distributed Algorithms}, 
            pdfauthor={Fabian Reiter},
            pdfkeywords={Graphs, Finite automata, MSO-logic, Distributed algorithms, Verification},
            linkbordercolor=skyblue, citebordercolor=skyblue, urlbordercolor=skyblue}
\usepackage[capitalise]{cleveref}
\crefname{subsection}{subsection}{subsections}

\usepackage{sty/commands}
\renewcommand{\defd}[1]{%
  \ifmmode {\color{darkred} #1}%
  \else {\em\color{darkred} #1}%
  \fi}

\usepackage{sty/algoconfig}

\newcommand{\TODO}[1][]{\textsf{\color{RawSienna}[TODO\ifthenelse{\isempty{#1}}{}{: #1}]}}

\newcommand{\UnCo}{{\ast\hspace{-.2ex}}}  
\newcommand{\Wp}{{\operatorname{wp}}}
\newcommand{\ini}{{\operatorname{ini}}}
\newcommand{\old}{{\operatorname{old}}}

\newcommand{\Var}{\mathit{Var}}
\newcommand{\Val}{\mathit{Val}}
\DeclareMathAlphabet{\mathsfit}{\encodingdefault}{\sfdefault}{m}{sl}
\newcommand{\NT}[1]{\mathsfit{#1}}  

\begin{document}

\title{\bf Distributed Graph Automata \\
  and \\
  Verification of Distributed Algorithms}
\author{Fabian Reiter \\
  \href{mailto:fabian.reiter@gmail.com}{\nolinkurl{fabian.reiter@gmail.com}}}
\date{May 2014}
\maketitle

\renewcommand{\abstractname}{\vspace{-\baselineskip}}
\begin{abstract}
  \noindent\textbf{Abstract.}\:\!
  Combining ideas from distributed algorithms and alternating
  automata, we introduce a new class of finite graph automata that
  recognize precisely the languages of finite graphs definable in
  monadic second-order logic. By restricting transitions to be
  nondeterministic or deterministic, we also obtain two strictly
  weaker variants of our automata for which the emptiness problem is
  decidable. As an application, we suggest how suitable graph automata
  might be useful in formal verification of distributed algorithms,
  using Floyd-Hoare logic.
  
  \bigskip\smallskip
  \noindent\textbf{Keywords.}\:\!
  Graphs, Finite automata, MSO-logic, Distributed algorithms, Verification
\end{abstract}

\vspace{5ex}
\tableofcontents

\newpage
\input{sec/introduction.tex}
\input{sec/preliminaries.tex}
\newpage
\input{sec/DGAs.tex}
\input{sec/DA_verification.tex}
\section*{Acknowledgments}
The author would like to thank Fabian Kuhn and Andreas Podelski (both
from the University of Freiburg) for many comments and stimulating
discussions. Their combined expertise was especially helpful for the
part about verification of distributed algorithms.

\end{document}

%% file: fig/graph_templates.tex
\newcommand{\pentagraphPic}[8]{
  \begin{tikzpicture}[lgraph,#1]
    \node[lnode] (u) {#3};
    \node[lnode] (v) at ([shift={(90:#2)}]u) {#4};
    \node[lnode] (w) at ([shift={(18:#2)}]u) {#5};
    \node[lnode] (x) at ([shift={(306:#2)}]u) {#6};
    \node[lnode] (y) at ([shift={(234:#2)}]u) {#7};
    \node[lnode] (z) at ([shift={(162:#2)}]u) {#8};
    \node[lnode,draw=none,fill=none] (dummy) at ([shift={(270:#2)}]u) {};
    \path (u)     edge[bend left=20] (v)
          (v)     edge[bend left=20] (u)
                  edge (w)
          (w.295) edge[bend left=20] (x)
          (x)     edge (u)
          (x.75)  edge[bend left=20] (w.250)
          (y)     edge (u);
  \end{tikzpicture}
}

%% file: sec/introduction.tex
\section{Introduction} \label{sec:introduction}

The regularity of a language of finite words is a central notion in
formal language theory. It is often defined as being recognizable by a
finite automaton, but many alternative characterizations exist. By
several well-known results, mostly from the late 1950s and early
1960s, it is equivalent whether a language can be
\begin{enumerate}[label=(\alph*),topsep=1ex,itemsep=0ex]
\item \label{itm:automata} recognized by a (non)deterministic or
  alternating finite automaton \cite{RS59,CKS81},
\item expressed by a regular expression \cite{Kle56},
\item generated by a regular grammar \cite{Cho56},
\item \label{itm:algebraic} obtained as a homomorphic preimage of a
  subset of some finite monoid \cite{Ner58}, or
\item \label{itm:logic} defined in (existential) monadic second-order
  logic \cite{Buc60,Elg61,Tra61}.
\end{enumerate}
All of these characterizations can be generalized from words to trees
in a natural manner, and, quite remarkably, they all remain equivalent
on trees (see, e.g., \cite{TATA08}). Hence, the notion of regularity
extends directly to tree languages.

In contrast, the situation becomes far more complicated if we expand
our field of interest from words or trees to arbitrary finite graphs
(possibly with node labels and multiple edge relations). Although some
of the characterizations mentioned above can be generalized to graphs
in a meaningful way, they are, in general, no longer equivalent.
Perhaps the logical approach~\ref{itm:logic} is the most
straightforward to generalize, since the syntax of monadic
second-order logic (MSO-logic) on graphs remains essentially the same
as on more restricted structures. While on words and trees the
existential fragment of that logic (EMSO-logic) is already sufficient
to characterize regularity, it is strictly less expressive than full
MSO-logic on graphs, as has been shown by Fagin in \cite{Fag75}.
Similarly, the algebraic approach~\ref{itm:algebraic} has been
extended to graphs by Courcelle in \cite{Cou90}, and it turns out that
MSO-logic is strictly less powerful than his notion of
recognizability, which is defined in terms of homomorphisms into
finite algebras. A common pattern that emerges from such results is
that the different characterizations of regularity drift apart as the
complexity of the considered structures increases. In this sense,
regularity cannot be considered a well-defined property of graph
languages.

To complicate matters even further, the automata-theoretic
characterization~\ref{itm:automata}, which is instrumental in the
theory of word and tree languages, does not seem to have a natural
counterpart on graphs. A word or tree automaton can scan its entire
input by a single deterministic traversal, which is completely
determined by the structure of the input (i.e., left-to-right, for
words, or bottom-to-top, for trees). On arbitrary graphs, however,
there is no sense of a global direction that the automaton could
follow, especially since we do not even require connectivity or
acyclicity.

Another approach, investigated by Thomas in \cite{Tho91}, is to
nondeterministically assign a state of the automaton to each node of
the graph, and then check that this assignment satisfies certain local
“transition” conditions for each node (specified with respect to
neighboring nodes within a fixed radius) as well as certain global
occurrence conditions at the level of the entire graph. The graph
acceptors introduced by Thomas, following this principle, turn out to
be equivalent to EMSO-logic on graphs of bounded degree. They are a
legitimate generalization of finite automata, in the sense that they
are equivalent to them and can easily simulate them if we restrict the
input to (graphs representing) words or trees. However, on arbitrary
graphs, they are less well-behaved than classical finite automata,
which is a direct consequence of their equivalence with EMSO-logic. In
particular, they do not satisfy closure under complementation, and
their emptiness problem is undecidable.

\paragraph{Contribution.} In this paper, we attempt to provide an
alternative approach to automata theory on finite graphs. Our model,
dubbed \emph{distributed graph automaton}, takes inspiration from
distributed algorithms and shares some similarities with Thomas' graph
acceptors. More specifically, we also use a combination of local
conditions, which are checked by the nodes using information received
from their neighborhood, and global conditions, which are checked at
the level of the entire graph. However, both types of conditions are
much simpler than in Thomas' model, which allows us to consider graphs
of unbounded degree. Nevertheless, we obtain as a main result that our
automata are equivalent to full MSO-logic if we equip them with the
power of alternation. If, on the other hand, we only allow
nondeterminism, then we get a model that is not closed under
complementation, and is even strictly weaker than EMSO-logic, but has
a decidable emptiness problem. Interestingly, this model is still
powerful enough to characterize precisely the regular languages when
restricted to words or trees. Hence, this work also contributes to the
general observation, made above, that regularity becomes a moving
target when lifted to the setting of graphs. Lastly, by further
disallowing nondeterminism, we obtain an even weaker model of
computation, which we use to illustrate how automata theory on graphs
might have an application in formal verification of distributed
algorithms.

\paragraph{Structure.} The remainder of this article is organized as
follows: Some preliminaries on graphs and logic are reviewed in
\cref{sec:preliminaries}. Then we introduce the notion of distributed
graph automaton and present our results in \cref{sec:DGAs}. That
section is mostly self-contained and constitutes the main part of this
paper. Finally, in \cref{sec:DA_verification}, we sketch an adaptation
of Floyd-Hoare logic to synchronous distributed algorithms. Although
the idea is presented using (the deterministic variant of) our
automaton model, it can be generalized to any type of graph automaton
that satisfies certain properties.

%% file: sec/preliminaries.tex
\section{Preliminaries} \label{sec:preliminaries}
We begin by fixing the terminology and notation used in this paper.

\subsection{Graphs and Graph Languages}
Our objects of interest are finite directed graphs with nodes labeled
by an alphabet $Σ$, and multiple edge relations indexed by an alphabet
$Γ$.

\begin{definition}[$Σ$-Labeled $Γ$-Graph]
  Let $Σ$ and $Γ$ be two finite nonempty sets of node labels and edge
  labels, respectively. A \defd{$Γ$-graph} is a structure
  $G=\bigl\langle\VG,\,⟨\arrG{γ}⟩_{\:\!γ∈Γ}\bigr\rangle$, where
  \begin{itemize}
  \item $\VG$ is a finite nonempty set of \defd{nodes}, and
  \item each ${\arrG{γ}}⊆\VG×\VG$ is a set of directed \defd{edges}
    labeled by $γ∈Γ$.
  \end{itemize}
  A (node) \defd{labeling} of $G$ is a function $λ\colon \VG→Σ$. We
  call the tuple $⟨G,λ⟩$ a \defd{$Σ$-labeled} $Γ$-graph and denote it
  by $\defd{G_λ}$.
\end{definition}
If $Σ$ and $Γ$ are understood or irrelevant, we refer to $G$ simply as
a graph and to $G_λ$ as a labeled graph, or even just as a graph. We
do this especially when the alphabets contain only a single “dummy”
symbol, which by default shall be the blank symbol $\blank$. If
$Σ=\{\blank\}$, we also identify $G_λ$ with $G$.

Given a $Γ$-graph $G$, we denote by $Σ^G$ the set of all $Σ$-labeled
versions of $G$, and by $Σ^{\clouded{Γ}}$ (read “$Σ$ clouded $Γ$”) the
set of all $Σ$-labeled $Γ$-graphs, i.e.,
\begin{equation*}
  \defd{Σ^G} ≔ \{G_λ \mid λ\colon\VG→Σ\},\!
  \quad \text{and} \quad
  \defd{Σ^{\clouded{Γ}}} ≔ \,\smashoperator{\bigcup_{G∈\mathcal{G}(Γ)}}\, Σ^G,
\end{equation*}
where $\mathcal{G}(Γ)$ is the set of all $Γ$-graphs. Note that this is
very similar to the standard notation of formal language theory on
words, where $Σ^n$ designates the set of all $Σ$-labeled versions of a
path of length $n$ (i.e., words over $Σ$ of length $n$), and
$Σ^*=\bigcup_{n∈ℕ}Σ^n$.

We are only interested in (labeled) graphs \emph{up to
  isomorphism}. That is, we consider $G_λ,G'_{λ'}∈Σ^{\clouded{Γ}}$ to
be equal if there is a bijection between $\VG$ and $\VGpr$ that
preserves the edge relations and node labels.

A \defd{graph language} is a set of labeled graphs. More precisely,
$L$ is a graph language \Iff there are finite nonempty alphabets $Σ$
and $Γ$, such that $L⊆Σ^{\clouded{Γ}}$.

By a (node) \defd{projection} we mean a mapping $h\colon Σ → Σ'$
between two alphabets $Σ$ and $Σ'$. With slight abuse of notation,
such a mapping is extended to labeled graphs by applying it to each
node label, and to graph languages by applying it to each labeled
graph. That is, for every $G_λ∈Σ^{\clouded{Γ}}$ and
$L⊆Σ^{\clouded{Γ}}$,
\begin{equation*}
  \defd{h(G_λ)} \coloneqq G_{h∘λ},\!
  \quad \text{and} \quad  
  \defd{h(L)} \coloneqq \{h(G_λ)\mid G_λ∈L\},
\end{equation*}
where the operator $∘$ denotes function composition, such that
$(h∘λ)(v)=h(λ(v))$.

When reasoning about graphs as structural objects, we will follow the
usual terminology of graph theory. In particular, given a $Γ$-graph
$G$ and two nodes $u,v∈\VG$, we say that $u$ is an \defd{incoming
  neighbor} of $v$, and $v$ an \defd{outgoing neighbor} of $u$, if
$u\arrG{γ}v$ for some $γ∈Γ$. In this case we also say that $u$ and $v$
are \defd{adjacent}, and without further qualification the term
\defd{neighbor} refers to both incoming and outgoing neighbors. The
\defd{neighborhood} of a node is the set of all of its neighbors. A
node without incoming neighbors is called a \defd{source}, whereas a
node without outgoing neighbors is called a \defd{sink}.

Finally, let us briefly recall some standard graph
properties. Consider a graph $G_λ∈Σ^{\clouded{Γ}}$. We say that $G_λ$
is \defd{undirected} if for every $u,v∈\VG$ and $γ∈Γ$, it holds that
$u\arrG{γ}v$ \Iff $v\arrG{γ}u$. The graph $G_λ$ is (weakly)
\defd{connected} if for every nonempty proper subset $U$ of $\VG$,
there exist two nodes $u∈U$ and $v∈\VG\setminus U$ that are
adjacent. The node labeling $λ$ constitutes a valid \defd{coloring} of
$G$ if no two adjacent nodes share the same label, i.e.,
\noheight{$u\arrG{γ}v$} implies $λ(u)≠λ(v)$, for all $u,v∈\VG$ and
$γ∈Γ$. If $\card{Σ}=k$, such a coloring is called a $k$-coloring of
$G$, and any $Γ$-graph for which a $k$-coloring exists is said to be
\defd{$k$-colorable}. Note that, by definition, a graph that contains
self-loops is not $k$-colorable for any $k$.

\subsection{Logic on Graphs}
We fix two disjoint, countably infinite sets of variables: the supply
of node variables
$\defd{\Vnode}=\{\lsymb{u},\lsymb{v},…,\lsymb{u_1},…\}$, and the
supply of set variables
$\defd{\Vset}=\{\lsymb{U},\lsymb{V},…,\lsymb{U_1},…\}$. Node variables
will always be represented by lower-case letters, and set variables by
upper-case ones, often with subscripts.

\enlargethispage{4ex}
\begin{definition}[Monadic Second-Order Formula] \label{def:mso-formula}
  Let $Σ$ and $Γ$ be two finite nonempty alphabets. The set
  \defd{$\MSO(Σ,Γ)$} of \defd{monadic second-order formulas} (on
  graphs) over $⟨Σ,Γ⟩$ is built up from the atomic formulas
  \begin{itemize}[topsep=1ex,itemsep=0ex]
  \item $\swl{\logic{\lab{\meta{a}}\meta{x}}}{\logic{\meta{x}\xarr{\meta{γ}}\meta{y}}}$
    \quad (“$x$ has label $a$”),
  \item $\logic{\meta{x}\xarr{\meta{γ}}\meta{y}}$
    \quad (“$x$ has a $γ$-edge to $y$”),
  \item $\swl{\logic{\meta{x}=\meta{y}}}{\logic{\meta{x}\xarr{\meta{γ}}\meta{y}}}$
    \quad (“$x$ is equal to $y$”),
  \item $\swl{\logic{\meta{x}∈\meta{X}}}{\logic{\meta{x}\xarr{\meta{γ}}\meta{y}}}$
    \quad (“$x$ is an element of $X$”),
  \end{itemize}
  for all\, $x,y\!∈\!\Vnode$,\, $X\!∈\!\Vset$,\, $a\!∈\!Σ$,\, and
  $γ\!∈\!Γ$, using the usual propositional connectives and
  quantifiers, which can be applied to both node and set
  variables. More precisely, if $φ$ and $ψ$ are $\MSO(Σ,Γ)$-formulas,
  then so are $\logic{¬\meta{φ}}$,\, $\logic{\meta{φ}∨\meta{ψ}}$,\,
  $\logic{\meta{φ}∧\meta{ψ}}$,\, $\logic{\meta{φ}⇒\meta{ψ}}$,\,
  $\logic{\meta{φ}⇔\meta{ψ}}$,\, $\logic{∃\meta{x}(\meta{φ})}$,\,
  $\logic{∀\meta{x}(\meta{φ})}$,\, $\logic{∃\meta{X}(\meta{φ})}$, and
  $\logic{∀\meta{X}(\meta{φ})}$, for all $x∈\Vnode$ and $X∈\Vset$.
\end{definition}

We denote by $\defd{\free(φ)}$ the set of variables in $\Vnode∪\Vset$
that occur freely in $φ$ (i.e., not within the scope of a quantifier),
and use the notation \defd{$φ[x_1,…,x_m,X_1,…,X_n]$} to indicate that
at most the variables given in brackets occur freely in $φ$, i.e.,
$\free(φ)⊆\{x_1,…,x_m,X_1,…,X_n\}$. If\, $\free(φ)=∅$, we also say
that $φ$ is a \defd{sentence}.

The truth of an $\MSO(Σ,Γ)$-formula $φ$ is evaluated with respect to a
labeled graph $G_λ∈Σ^{\clouded{Γ}}$ and a variable assignment\,
$α\colon \free(φ)→\VG ∪ 2^\VG$ that assigns a node $v∈\VG$ to each
node variable in $\free(φ)$, and a set of nodes $S⊆\VG$ to each set
variable in $\free(φ)$. The meaning of atomic formulas is as hinted
informally in \cref{def:mso-formula}. In particular,
\noheight{$\logic{\lab{\meta{a}}\meta{x}}$} is satisfied \Iff
$λ(α(x))=a$, and \noheight{$\logic{\meta{x}\xarr{\meta{γ}}\meta{y}}$}
is satisfied \Iff $α(x)\arrG{γ}α(y)$. For composed formulas,
satisfaction is defined inductively by the standard semantics of
predicate logic. We write $\defd{⟨G_λ,α⟩⊨φ}$ to denote that $G_λ$ and
$α$ \defd{satisfy} $φ$. If $φ$ is a sentence, the variable assignment
is superfluous, and we simply write $\defd{G_λ⊨φ}$ if $G_λ$ satisfies
$φ$.

The graph language $\Lf{Σ,Γ}(φ)$ \defd{defined} by $φ$ with respect to
$Σ$ and $Γ$ is the set of all $Σ$-labeled $Γ$-graphs that satisfy $φ$,
i.e.,
\begin{equation*}
  \defd{\Lf{Σ,Γ}(φ)} \coloneqq \bigl\{ G_λ∈Σ^{\clouded{Γ}} \bigm| G_λ⊨φ \bigr\}.
\end{equation*}
Every graph language that is defined by some MSO-sentence is called
\defd{MSO-definable}. We denote by \defd{$\LL_\MSO$} the class of all
such graph languages.

\begin{example}[3-Colorability] \label{ex:MSO_3_colorable}
  Let $Σ=Γ=\{\blank\}$. The following $\MSO(Σ,Γ)$-sentence defines the
  language of 3-colorable graphs.
  \vspace{-1ex}
  \begin{align*}
    \dphi[color]{3} \coloneqq \logic{∃U_\mpik,U_\mherz,U_\mkreuz\biggl(}
    &\logic{∀u\Bigl(\bigl(u∈U_\mpik ∨ u∈U_\mherz ∨ u∈U_\mkreuz\bigr)
      \,∧\, ¬\bigl(u∈U_\mpik ∧ u∈U_\mherz\bigr) \,∧ \vphantom{\biggl(} } \hspace{-5ex} \\[-1.5ex]
    &\hspace{5ex}\logic{¬\bigl(u∈U_\mpik ∧ u∈U_\mkreuz\bigr)
      \,∧\, ¬\bigl(u∈U_\mherz ∧ u∈U_\mkreuz\bigr)\:\!\Big) \;∧} \\[-.9ex]
    &\logic{∀u,v\Big(u\!\arr\!v
      \;\,⇒\;\, ¬\bigl(u∈U_\mpik∧v∈U_\mpik\bigr) \,∧} \\[-1.5ex]
    &\hspace{11.2ex}\logic{¬\bigl(u∈U_\mherz∧v∈U_\mherz\bigr)
      \,∧\, ¬\bigl(u∈U_\mkreuz∧v∈U_\mkreuz\bigr)\:\!\Bigr)\,\biggr)} \hspace{-5ex}
  \end{align*}
  The existentially quantified set variables $\lsymb{U_\mpik}$,
  $\lsymb{U_\mherz}$ and $\lsymb{U_\mkreuz}$ represent the three
  possible colors. In the first two lines, we specify that the sets
  assigned to these variables form a partition of the set of nodes
  (possibly with empty components). The remaining two lines constitute
  the actual definition of a valid coloring: no two adjacent nodes
  share the same color, which means that adjacent nodes are in
  different sets.
\end{example}

A \defd{first-order formula} (FO-formula) is an MSO-formula in which
set variables may not be bound by quantifiers, i.e., subformulas of
the form \noheight{$\logic{∃\meta{X}(\meta{φ})}$} and
\noheight{$\logic{∀\meta{X}(\meta{φ})}$} are disallowed, for
$X∈\Vset$. An \defd{existential MSO-formula} (EMSO-formula) is of the
form
\noheight{$\logic{∃\meta{X_\one},\meta{…}\,,\meta{X_n}(\meta{φ})}$},
where $X_1,…,X_n∈\Vset$ and $φ$ is an FO-formula. We denote the
classes of FO- and EMSO-definable graph languages by $\defd{\LL_\FO}$
and $\defd{\LL_\EMSO}$. (Note that by \cref{ex:MSO_3_colorable}, the
language of 3-colorable graphs lies in $\LL_\EMSO$.)

%% file: sec/DGAs.tex
\section{Distributed Graph Automata} \label{sec:DGAs}
The simple idea of interconnecting finite-state machines in a
synchronous distributed setting presents a natural paradigm for
defining finite automata on graphs of arbitrary topology. In this
section, we introduce three classes of automata obtained this way, and
discuss some of their properties. Our most powerful version of
distributed graph automata turns out to be equivalent to MSO-logic on
graphs. The other two are restricted variants for which the emptiness
problem is decidable.

\subsection{Informal Description} \label{sec:dga-preview}
We start with an informal description of our automaton model. Formal
definitions follow in \cref{sec:dga-definitions}.

A distributed graph automaton (DGA) is an abstract machine that, given
a labeled graph as input, can either accept or reject it, thereby
specifying a graph language. Our model of computation incorporates the
following key concepts:
\paragraph{Synchronous Distributed Algorithm.} A DGA operates
primarily as a distributed algorithm. Each node of the input graph is
assigned its own local processor, which we shall not distinguish from
the node itself. Communication takes place in synchronous rounds, in
which each node receives the current states of its incoming neighbors.
\paragraph{Finite-State Machines.} Each local processor is a
finite-state machine, i.e., an abstract machine that can be in one of
a finite number of states, and has no additional memory. Its initial
state is determined by the node label. After each communication round,
it updates its state according to a (possibly nondeterministic)
transition function that depends only on the current state and the
states received from the incoming neighborhood.
\paragraph{Constant Running Time.} The number of communication rounds
is limited by a constant. To ensure this, we associate a number,
called \emph{level}, with every state. In most cases, this number
indicates the round in which the state may occur. We require that
potentially initial states are at level $0$, and outgoing transitions
from states at level $i$ go to states at level $i+1$. There is an
exception, however: the states at the highest level, called the
\emph{permanent states}, can also be initial states, and can have
incoming transitions from any level. Moreover, all their outgoing
transitions are self-loops. The idea is that, once a node has reached
a permanent state, it terminates its local computation, and waits for
the other nodes in the graph to terminate too.
\paragraph{Aggregation of States.} In order to be finitely
representable, a DGA treats collections of states as sets, i.e., it
abstracts away from the multiplicity of states. This aggregation of
states into sets is applied in two ways:
\begin{itemize}
\item First, the information received by the nodes in each round is a
  family of sets of states, indexed by the edge alphabet of the
  graph. That is, for each edge relation, a node knows which states
  occur in its incoming neighborhood, but it cannot distinguish
  between neighbors that are in the same state.
\item Second, once all the nodes have reached a permanent state, the
  DGA ceases to operate as a distributed algorithm, and collects all
  the reached permanent states into a set $F$. This set is the sole
  acceptance criterion: if $F$ is part of the DGA's accepting sets,
  then the input graph is accepted, otherwise it is rejected.
\end{itemize}

As an introductory example, let us translate the MSO-formula
$\dphi[color]{3}$ from \cref{ex:MSO_3_colorable} to the setting of
DGAs.

\begin{example}[3-Colorability] \label{ex:ADGA_3_colorable}
  \Cref{fig:ADGA_3_colorable} shows the state diagram of a simple
  nondeterministic DGA $\dA[color]{3}$\!. The states are arranged in
  columns corresponding to their levels, ascending from left to right.
  $\dA[color]{3}$ expects a $\{\blank\}$-labeled $\{\blank\}$-graph as
  input, and accepts it \Iff it is 3-colorable. The automaton proceeds
  as follows: All nodes of the input graph are initialized to the
  state $\q{ini}$. In the first round, each node nondeterministically
  chooses to go to one of the states $q_\pik$, $q_\herz$ and
  $q_\kreuz$, which represent the three possible colors. Then, in the
  second round, the nodes verify locally that the chosen coloring is
  valid. If the set received from their incoming neighborhood (only
  one, since there is only a single edge relation) contains their own
  state, they go to $\q{no}$, otherwise to $\q{yes}$. The automaton
  then accepts the input graph \Iff all the nodes are in $\q{yes}$,
  i.e., $\{\q{yes}\}$ is its only accepting set. This is indicated by
  the blue bar to the right of the state diagram. We shall refer to
  such a representation of sets using bars as \emph{barcode}.
\end{example}

\begin{figure}[h!]
  \alignpic
  \input{fig/ADGA_3_colorable.tikz}
  \caption{$\dA[color]{3}$\!, a nondeterministic DGA over
    $\bigl\langle\{\blank\},\{\blank\}\bigr\rangle$ whose graph
    language consists of the 3-colorable graphs.}
  \label{fig:ADGA_3_colorable}
\end{figure}

One last key concept that enters into our most general definition of
DGAs is \emph{alternation}, a generalization of nondeterminism
introduced by Chandra, Kozen and Stockmeyer in \cite{CKS81} (in their
case, for Turing machines and other types of word automata).
\paragraph{Alternating Automaton.} In addition to being able to
nondeterministically choose between different transitions, nodes can
also explore several choices in parallel. To this end, the
nonpermanent states of an alternating DGA (ADGA) are partitioned into
two types, \emph{existential} and \emph{universal}, such that states
on the same level are of the same type. If, in a given round, the
nodes are in existential states, then they nondeterministically choose
a single state to go to in the next round, as described above. In
contrast, if they are in universal states, then the run of the ADGA is
split into several parallel branches, called universal branches, one
for each possible combination of choices of the nodes. This procedure
of splitting is repeated recursively for each round in which the nodes
are in universal states. The ADGA then accepts the input graph \Iff
its acceptance condition is satisfied in every universal branch of the
run.

\begin{example}[Non-3-Colorability]
  To illustrate the notion of universal branching, consider the ADGA
  $\dcA[color]{3}$ shown in \cref{fig:ADGA_not_3_colorable}. It is a
  complement automaton of $\dA[color]{3}$ from
  Example~\ref{ex:ADGA_3_colorable}, i.e., it accepts precisely those
  $\{\blank\}$-labeled $\{\blank\}$-graphs that are \emph{not}
  3-colorable. States represented as red triangles are universal
  (whereas the green squares in \cref{fig:ADGA_3_colorable} stand for
  existential states). Given an input graph with $n$ nodes,
  $\dcA[color]{3}$ proceeds as follows: All nodes are initialized to
  $\q{ini}$. In the first round, the run is split into $3^n$ universal
  branches, each of which corresponds to one possible outcome of the
  first round of $\dA[color]{3}$ running on the same input
  graph. Then, in the second round, in each of the $3^n$ universal
  branches, the nodes check whether the coloring chosen in that branch
  is valid. As indicated by the barcode, the acceptance condition of
  $\dcA[color]{3}$ is satisfied \Iff at least one node is in state
  $\q{no}$, i.e., the accepting sets are $\{\q{no}\}$ and
  $\{\q{yes},\q{no}\}$. Hence, the automaton accepts the input graph
  \Iff no valid coloring was found in any universal branch. Note that
  we could also have chosen to make the states $q_\pik$, $q_\herz$ and
  $q_\kreuz$ existential, since their outgoing transitions are
  deterministic. Regardless of their type, there is no branching in
  the second round.
\end{example}

\begin{figure}[h!]
  \alignpic
  \input{fig/ADGA_not_3_colorable.tikz}
  \caption{$\dcA[color]{3}$\!, an alternating DGA over
    $\bigl\langle\{\blank\},\{\blank\}\bigr\rangle$ whose graph
    language consists of the graphs that are \emph{not} 3-colorable.}
  \label{fig:ADGA_not_3_colorable}
\end{figure}

\subsection{Formal Definitions} \label{sec:dga-definitions}
We now repeat and clarify the notions from \cref{sec:dga-preview} in a
more formal setting, beginning with our most general definition of
DGAs.

\enlargethispage*{3ex}
\begin{definition}[Alternating Distributed Graph Automaton] \label{def:adga}
  An \defd{alternating distributed graph automaton} (ADGA) over
  alphabets $⟨Σ,Γ⟩$ is a tuple $\A=⟨Σ,Γ,\Q,\ab σ,δ,\F⟩$, where
  \begin{itemize}
  \item $Σ$ and $Γ$ are finite nonempty sets of node labels and edge
    labels, respectively,
  \item $\Q=⟨Q_\EE,Q_\AA,Q_\P⟩$, where $Q_\EE$, $Q_\AA$ and $Q_\P$,
    with $Q_\P≠∅$, are pairwise disjoint finite sets of
    \defd{existential}, \defd{universal} and \defd{permanent} states,
    respectively, which are also referred to by the notational
    shorthands
    \begin{itemize}[topsep=0ex,itemsep=0ex]
    \item $\swl{Q}{Q_\N} \coloneqq Q_\EE∪Q_\AA∪Q_\P,$\, for the entire
      set of \defd{states},
    \item $Q_\N \coloneqq \swl{Q_\EE∪Q_\AA,}{Q_\EE∪Q_\AA∪Q_\P,}$\, for
      the set of \defd{nonpermanent} states,
    \end{itemize}
  \item $σ\colon Σ→Q$ is an \defd{initialization function},
  \item $δ\colon Q×(2^Q)^Γ→2^Q$ is a (local) \defd{transition
      function},\, and
  \item $\F⊆2^{Q_\P}$ is a set of \defd{accepting sets} of permanent
    states.
  \end{itemize}
  The functions $σ$ and $δ$ must be such that one can unambiguously
  associate with every state $q∈Q$ a \defd{level} $\defd{\lA(q)}∈ℕ$
  satisfying the following conditions:
  \begin{itemize}[itemsep=0ex]
    \setlength\abovedisplayskip{1ex}
    \setlength\belowdisplayskip{1ex}
  \item States on the same level are of the same type, i.e., for every
    $i∈ℕ$,
    \begin{equation*}
      \{q∈Q \mid \lA(q)=i\}∈(2^{Q_\EE}∪2^{Q_\AA}∪2^{Q_\P}).
    \end{equation*}
  \item Initial states are either on the lowest level or permanent,
    i.e., for every $q∈Q$,
    \begin{equation*}
      ∃a∈Σ\colon σ(a)=q \quad \text{implies} \quad \lA(q)=0 \,∨\, q∈Q_\P.
    \end{equation*}
  \item Nonpermanent states without incoming transitions are on the
    lowest level, and transitions between nonpermanent states go only
    from one level to the next, i.e., for every $q∈Q_\N$,
    \begin{equation*}
      \lA(q)=
      \begin{cases}
        0 & \text{if for all $p∈Q$ and $\S∈(2^Q)^Γ$\!,
          it holds that $q∉δ(p,\S)$,} \\[.2ex]
        i+1 & \text{\parbox[t]{50ex}{if there are $p∈Q_\N$ and
            $\S∈(2^Q)^Γ$ such that $\lA(p)=i$ and $q∈δ(p,\S)$.}}
      \end{cases}
    \end{equation*}
  \item The permanent states are one level higher than the highest
    nonpermanent ones, and have only self-loops as outgoing
    transitions, i.e., for every $q∈Q_\P$,
    \begin{gather*}
      \lA(q)=
      \begin{cases}
        0                          & \text{if\; $Q_\N=∅$}, \\
        \max\{\lA(q)\mid q∈Q_\N\}+1 & \text{otherwise},
      \end{cases} \\[.2ex]
      δ(q,\S)=\{q\} \quad \text{for every $\S∈(2^Q)^Γ$}\!.
    \end{gather*}
  \end{itemize}
\end{definition}

For any ADGA $\A=⟨Σ,Γ,\ab\Q,\ab σ,\ab δ,\F⟩$, we define its
\defd{length} $\len(\A)$ to be its highest level, i.e.,
$\defd{\len(\A)} \coloneqq \max\{\lA(q)\mid q∈Q\}$.

Next, we want to give a formal definition of a run. For this, we need
the notion of a configuration, which can be seen as the global state
of an ADGA.

\begin{definition}[Configuration]
  ~ \\[-2.5ex]
  Consider an ADGA $\A=⟨Σ,Γ,\ab\Q,\ab σ,\ab δ,\F⟩$. We call any
  $Q$-labeled $Γ$-graph $G_κ∈Q^{\clouded{Γ}}$ a \defd{configuration}
  of $\A$ on $G$. If every node in $G_κ$ is labeled by a permanent
  state, we refer to $G_κ$ as a \defd{permanent}
  configuration. Otherwise, if $G_κ$ is a nonpermanent configuration
  whose nodes are labeled exclusively by existential and (possibly)
  permanent states, we say that $G_κ$ is an \defd{existential}
  configuration. Analogously, $G_κ$ is \defd{universal} if it is
  nonpermanent and only labeled by universal and (possibly) permanent
  states.

  Additionally, we say that a permanent configuration $G_κ$ is
  \defd{accepting} if the set of states occurring in it is accepting,
  i.e., if $\{κ(v) \mid v∈\VG\}∈\F$. Any other permanent configuration
  is called \defd{rejecting}. Nonpermanent configurations are neither
  accepting nor rejecting.
\end{definition}

The (local) transition function of an ADGA specifies for each state a
set of potential successors, for a given family of sets of
states. This can be naturally extended to configurations, which leads
us to the definition of a global transition function.

\begin{definition}[Global Transition Function]
  The \defd{global transition function} $δ^\cloud$ of an ADGA
  $\A=⟨Σ,Γ,\ab\Q,\ab σ,\ab δ,\F⟩$ assigns to each configuration $G_κ$
  of $\A$ the set of all of its \defd{successor configurations} $G_μ$,
  by combining all possible outcomes of local transitions on $G_κ$,
  i.e.,

  \begin{align*}
    δ^\cloud \colon Q^{\clouded{Γ}} &→ 2^{(Q^{\clouded{Γ}})} \\
    G_κ &↦ \biggl\{G_μ \biggm|
    \bigwedge_{v∈\VG} μ(v)∈ δ\Bigl(κ(v),\,\bigl\langle\{κ(u)\mid u\arrG{γ} v\}\bigr\rangle_{γ∈Γ}\Bigr)\biggr\}.
  \end{align*}
\end{definition}

\smallskip
We now have everything at hand to formalize the notion of a run.

\begin{definition}[Run]
  A \defd{run} of an ADGA $\A=⟨Σ,Γ,\Q,σ,δ,\F⟩$ on a labeled graph
  $G_λ∈Σ^{\clouded{Γ}}$ is a directed acyclic graph $R=⟨K,\arr⟩$ whose
  nodes are configurations of $\A$ on $G$, such that
  \begin{itemize}
  \item the \defd{initial configuration} $G_{σ∘λ}∈K$ is the only
    source,\tablefootnote{As before, the operator $∘$ denotes function
      composition, such that $(σ∘λ)(v)=σ(λ(v))$.}
  \item every nonpermanent configuration $G_κ\mathbin{∈}K$ with
    $δ^\cloud(G_κ)=\{G_{μ_1},…,G_{μ_m}\}$ has
    \begin{itemize}[topsep=0ex,itemsep=0ex]
    \item exactly one outgoing neighbor $G_{μ_i}∈δ^\cloud(G_κ)$\, if
      $G_κ$ is existential,
    \item exactly $m$ outgoing neighbors $G_{μ_1},…,G_{μ_m}$\, if
      $G_κ$ is universal,\, and
    \end{itemize}
  \item every permanent configuration $G_κ∈K$ is a sink.
  \end{itemize}
  The run $R$ is \defd{accepting} if every permanent configuration
  $G_κ∈K$ is accepting.
\end{definition}

An ADGA $\A=⟨Σ,Γ,\Q,σ,δ,\F⟩$ \defd{accepts} a labeled graph
$G_λ∈Σ^{\clouded{Γ}}$ \Iff there exists an accepting run $R$ of $\A$
on $G_λ$. The graph language \defd{recognized} by $\A$ is the set
\begin{equation*}
  \defd{\L(\A)} \coloneqq \bigl\{ G_λ∈Σ^{\clouded{Γ}} \bigm| \text{$\A$ accepts $G_λ$} \bigr\}.
\end{equation*}
Every graph language that is recognized by some ADGA is called
\defd{ADGA-recognizable}. We denote by \defd{$\LL_\ADGA$} the class of
all such graph languages.

The ADGA $\A$ is \defd{equivalent} to some $\MSO(Σ,Γ)$-sentence $φ$ if
it recognizes precisely the graph language defined by $φ$, i.e., if
$\L(\A)=\Lf{Σ,Γ}(φ)$.

We inductively define that a configuration $G_κ∈Q^{\clouded{Γ}}$ is
\defd{reachable} by $\A$ on $G_λ$ if either $G_κ=G_{σ∘λ}$, or
$G_κ∈δ^\cloud(G_μ)$ for some configuration $G_μ∈Q^{\clouded{Γ}}$
reachable by $\A$ on $G_λ$. In case $G_λ$ is irrelevant, we simply say
that $G_κ$ is reachable by $\A$.

The automaton $\A$ is called a \defd{nondeterministic DGA} (NDGA) if
it has no universal states, i.e., if $Q_\AA=∅$. If additionally every
configuration $G_κ∈Q^{\clouded{Γ}}$ that is reachable by $\A$ has
precisely one successor configuration, i.e., $\card{δ^\cloud(G_κ)}=1$,
then we refer to $\A$ as a \defd{deterministic DGA} (DDGA). We denote
the classes of NDGA- and DDGA-recognizable graph languages by
$\defd{\LL_\NDGA}$ and $\defd{\LL_\DDGA}$.

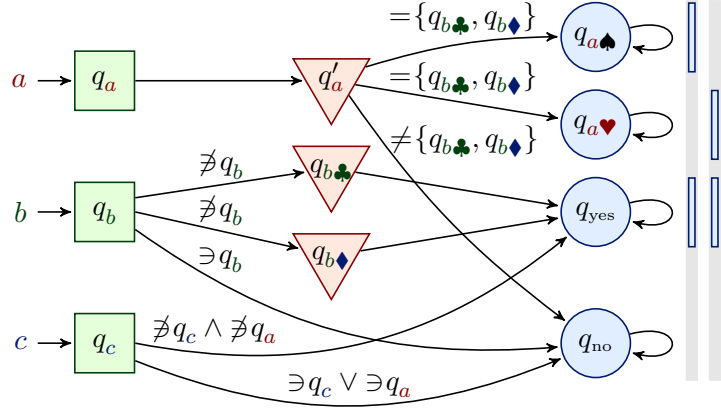
\begin{figure}[p]
  \alignpic
  \input{fig/ADGA_concentric_circles.tikz}
  \caption{$\sA{centric}$, an ADGA over
    $\bigl\langle\{\a,\b,\c\},\{\blank\}\bigr\rangle$ whose graph
    language consists of the labeled graphs that satisfy the following
    conditions: the labeling constitutes a valid 3-coloring, there is
    precisely one $\a$-labeled node $v_\a$, the undirected
    neighborhood of $v_\a$ contains only $\b$-labeled nodes, and
    $v_\a$ has at least two incoming neighbors.}
  \label{fig:ADGA_concentric_circles}
\end{figure}

\begin{figure}[p]
  \alignpic
  \input{fig/graph_labeled_pentagon.tikz}
  \caption{An $\{\a,\b,\c\}$-labeled $\{\blank\}$-graph.}
  \label{fig:graph_labeled_pentagon}
\end{figure}

\begin{figure}[p]
  \alignpic
  \input{fig/run_accepting.tikz}
  \caption{An accepting run of the ADGA of
    \cref{fig:ADGA_concentric_circles} on the labeled graph of
    \cref{fig:graph_labeled_pentagon}.}
  \label{fig:run_accepting}
\end{figure}
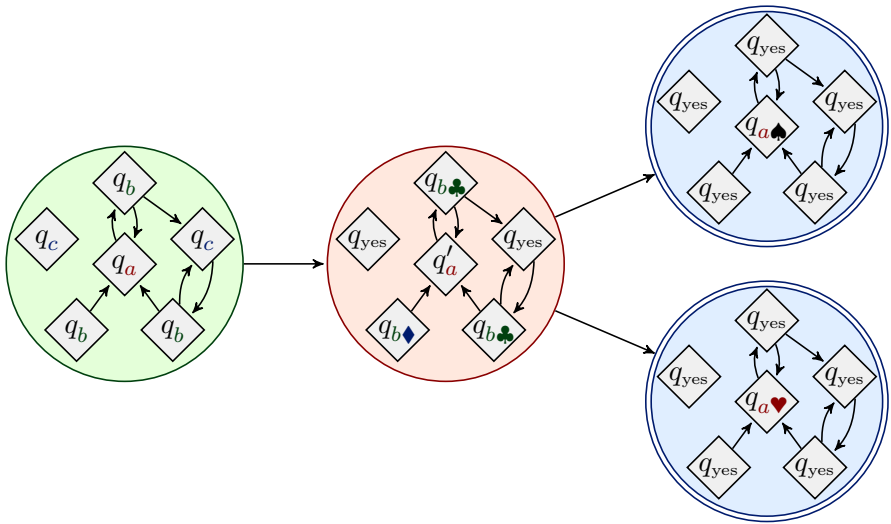

\bigskip
Let us now illustrate the notion of ADGA by means of a slightly more
involved example.

\begin{example}[Concentric Circles]
  Consider the ADGA $\sA{centric}=⟨Σ,Γ,\Q,σ,δ,\F⟩$ represented by the
  state diagram in \cref{fig:ADGA_concentric_circles}. The node and
  edge alphabets are $Σ=\{\a,\b,\c\}$ and $Γ=\{\blank\}$. Again,
  existential states are represented by green squares, universal
  states by red triangles, and permanent states by blue circles. The
  short arrows mapping node labels to states indicate the
  initialization function $σ$. For instance, $σ(\a)=\qa$. The other
  arrows specify the transition function $δ$. A label on such a
  transition arrow indicates a requirement on the set of states that a
  node receives from its incoming neighborhood (only one set, since
  there is only a single edge relation). For instance,
  $δ\bigl(\qb,⟨\{\qa,\qc\}⟩\bigl)=\{\qbkr,\qbk\}$. If there is no
  label, any set is permitted. Finally, as indicated by the barcode on
  the far right, the set of accepting sets is
  $\F=\bigl\{\{\qap,\q{yes}\},\{\qah,\q{yes}\}\bigr\}$.

  Intuitively, $\sA{centric}$ proceeds as follows: In the first round,
  the $\a$-labeled nodes do nothing but update their state, while the
  $\b$- and $\c$-labeled nodes verify that the labels in their
  incoming neighborhood satisfy the condition of a valid graph
  coloring. The $\c$-labeled nodes additionally check that they do not
  see any $\a$'s, and then directly terminate. Meanwhile, the
  $\b$-labeled nodes nondeterministically choose one of the markers
  $\kreuz$ and $\karo$. In the second round, only the $\a$-labeled
  nodes are busy. They verify that their incoming neighborhood
  consists exclusively of $\b$-labeled nodes, and that both of the
  markers $\kreuz$ and $\karo$ are present, thus ensuring that they
  have at least two incoming neighbors. Then, they simultaneously pick
  the markers $\pik$ and $\herz$, thereby creating different universal
  branches, and the run of the automaton terminates. Finally, the ADGA
  checks that all the nodes approve of the graph (meaning that none of
  them has reached the state $\q{no}$), and that in each universal
  branch, precisely one of the markers $\pik$ and $\herz$ occurs,
  which implies that there is a unique $\a$-labeled node.

  To sum up, the graph language $\L(\sA{centric})$ consists of all the
  $\{\a,\b,\c\}$-labeled $\{\blank\}$-graphs such that
  \begin{itemize}
    \item the labeling constitutes a valid 3-coloring,
    \item there is precisely one $\a$-labeled node $v_\a$, and
    \item $v_\a$ has only $\b$-labeled nodes in its undirected
      neighborhood, and at least two incoming neighbors.
  \end{itemize}
  The name “$\sA{centric}$” refers to the fact that, in the (weakly)
  connected component of $v_\a$, the $\b$- and $\c$-labeled nodes form
  con\emph{centric} circles around $v_\a$, i.e., nodes at distance~1
  of $v_\a$ are labeled with $\b$, nodes at distance~2 (if existent)
  with $\c$, nodes at distance~3 (if existent) with $\b$, and so
  forth.

  \Cref{fig:graph_labeled_pentagon} shows an example of a labeled
  graph that lies in $\L(\sA{centric})$. A corresponding accepting run
  can be seen in \cref{fig:run_accepting}. We have adopted the same
  coloring scheme as for (automaton) states, i.e., a green
  configuration is existential, a red one is universal, and a blue one
  is permanent. In the first round, the three nodes that are in state
  $\qb$ have a nondeterministic choice between $\qbkr$ and
  $\qbk$. Hence, the second configuration is one of eight possible
  choices. The branching in the second round is due to the node in
  state $\qaprime$ which goes simultaneously to $\qap$ and $\qah$. In
  both branches, an accepting configuration is reached, since
  $\{\qap,\q{yes}\}$ and $\{\qah,\q{yes}\}$ are both accepting
  sets. Therefore, the entire run is accepting.
\end{example}

\bigskip 

In the following subsections (\ref{sec:hierarchy-closure},
\ref{sec:adga=mso} and \ref{sec:ndga-emptiness}), we derive our
results on the properties of DGAs. For more detailed proofs and
further examples of automata (recognizing, e.g., connected or cyclic
graphs), see \cite{Rei14}.

\subsection{Hierarchy and Closure Properties} \label{sec:hierarchy-closure}
\begin{lemma}[Closure Properties of $\LL_\ADGA$] \label{lem:adga-closure}
  The class $\LL_\ADGA$ of ADGA-recognizable graph languages is
  effectively closed under Boolean set operations and under
  projection.
\end{lemma}

\begin{proofsketch}
  As usual for alternating automata, complementation can be achieved
  by simply swapping the existential and universal states, and
  complementing the acceptance condition. That is, for an ADGA
  $\A=\bigl\langle Σ,Γ,⟨Q_\EE,Q_\AA,Q_\P⟩,σ,δ,\F \bigr\rangle$, a
  complement automaton is $\cA=\bigl\langle
  Σ,Γ,⟨Q_\AA,Q_\EE,Q_\P⟩,σ,δ,2^{Q_\P}\setminus\F \bigr\rangle$. This
  can be easily seen by associating a two-player game with $\A$ and
  any $Σ$-labeled $Γ$-graph $G_λ$. One player tries to come up with an
  accepting run of $\A$ on $G_λ$, whereas the other player seeks to
  find a (path to a) rejecting configuration in any run proposed by
  the adversary. The first player has a winning strategy \Iff $\A$
  accepts $G_λ$. (This game-theoretic characterization will be used
  and explained more extensively in the proof of \cref{thm:adga=mso}.)
  From this perspective, the construction of $\cA$ corresponds to
  interchanging the roles and winning conditions of the two players.

  For two ADGAs $\A_1$ and $\A_2$, we can effectively construct an
  ADGA $\A_∪$ that recognizes $\L(\A_1)∪\L(\A_2)$ by taking advantage
  of nondeterminism. The approach is, in principle, very similar to
  the corresponding construction for nondeterministic finite automata
  on words. In the first round of $\A_∪$, each node in the input graph
  nondeterministically and independently decides whether to behave
  like in $\A_1$ or in $\A_2$. If there is a consensus, then the run
  continues as it would in the unanimously chosen automaton $\A_j$,
  and it is accepting \Iff it corresponds to an accepting run of
  $\A_j$. Otherwise, a conflict is detected, either locally by
  adjacent nodes that have chosen different automata, or at the
  latest, when acceptance is checked globally (important for
  disconnected graphs), and in either case the run is rejecting. (Note
  that we have omitted some technicalities that ensure that the
  construction outlined above satisfies all the properties of an
  ADGA.)

  Closure under node projection is straightforward, again by
  exploiting nondeterminism. Given an ADGA $\A$ with node alphabet $Σ$
  and a projection $h\colon Σ → Σ'$, we can effectively construct an
  ADGA $\A'$ that recognizes $h(\L(\A))$ as follows: For every $b∈Σ'$,
  each node labeled with $b$ nondeterministically chooses a new label
  $a∈Σ$, such that $h(a)=b$. Then, the automaton $\A$ is simulated on
  that new input.
\end{proofsketch}

\begin{lemma}[$\LL_\NDGA⊂\LL_\ADGA$] \label{lem:ndga<adga}
  There are (infinitely many) ADGA-recognizable graph languages that
  are not NDGA-recognizable.
\end{lemma}

\begin{proof}
  Let $Σ=Γ=\{\blank\}$. For any constant $k≥1$, we consider the
  language $\dL[order]{≤k}$ of all graphs that have at most $k$ nodes,
  i.e., $\dL[order]{≤k}=\bigl\{G∈Σ^{\clouded{Γ}} \bigm| \card{\VG}≤k
  \bigr\}$. We can easily construct an ADGA that recognizes this graph
  language: In a universal branching, each node goes to $k+1$
  different states in parallel. The automaton accepts \Iff there is no
  branch in which the $k+1$ states occur all at once. Now, assume for
  sake of contradiction that $\dL[order]{≤k}$ is also recognized by some
  NDGA $\A$, and let $G$ be a graph with $k$ nodes. We construct a
  variant $G'$ of $G$ with $k+1$ nodes by duplicating some node $v$,
  together with all of its incoming and outgoing edges. Observe that
  any accepting run of $\A$ on $G$ can be extended to an accepting run
  on $G'$, where the copy of $v$ behaves exactly like $v$ in every
  round.
\end{proof}

\begin{lemma}[Closure Properties of $\LL_\NDGA$] \label{lem:ndga-closure}
  The class $\LL_\NDGA$ of NDGA-recognizable graph languages is
  effectively closed under union, intersection and projection, but
  \emph{not} closed under complementation.
\end{lemma}

\begin{proof}
  For union and projection, we simply use the same constructions as
  for ADGAs (see \cref{lem:adga-closure}).

  Intersection can be handled by a product construction, similar to
  the one for finite automata on words. Given two NDGAs $\A_1$ and
  $\A_2$, we construct an NDGA $\A_⊗$ that operates on the Cartesian
  product of the state sets of $\A_1$ and $\A_2$. It simulates the two
  automata simultaneously and accepts \Iff both of them reach an
  accepting configuration.

  To see that $\LL_\NDGA$ is not closed under complementation, we
  recall from the proof of \cref{lem:ndga<adga} that for any $k≥1$,
  the language $\dL[order]{≤k}$ of all graphs that have at most $k$
  nodes is not NDGA-recognizable. However, complementing the ADGA
  given for $\dL[order]{≤k}$ yields an NDGA that recognizes the
  complement language $\dL[order]{≥k+1}$.
\end{proof}

\begin{lemma}[$\LL_\DDGA⊂\LL_\NDGA$] \label{lem:ddga<ndga}
  There are (infinitely many) NDGA-recognizable graph languages that
  are not DDGA-recognizable.
\end{lemma}

\begin{proof}
  Let $k≥2$. As mentioned in the proof of \cref{lem:ndga-closure}, the
  language $\dL[order]{≥k}$ of all graphs that have at least $k$ nodes is
  NDGA-recognizable. To see that it is not DDGA-recognizable, consider
  (similarly to the proof of \cref{lem:ndga<adga}) a graph $G$ with
  $k-1$ nodes and a variant $G'$ with $k$ nodes obtained from $G$ by
  duplicating some node $v$, together with all of its incoming and
  outgoing edges. Given any DDGA $\A$, the determinism of $\A$
  guarantees that $v$ and its copy $v'$ behave the same way in the
  (unique) run of $\A$ on $G'$. Hence, if that run is accepting, so is
  the run on $G$.
\end{proof}

\begin{lemma}[Closure Properties of $\LL_\DDGA$] \label{lem:ddga-closure}
  The class $\LL_\DDGA$ of DDGA-recognizable graph languages is
  effectively closed under Boolean set operations, but \emph{not}
  closed under projection.
\end{lemma}

\begin{proof}
  To complement a DDGA, we can simply complement its set of accepting
  sets. The product construction for intersection of NDGAs mentioned
  in \cref{lem:ndga-closure} remains applicable when restricted to
  DDGAs.

  Closure under node projection does not hold because we can, for
  instance, construct a DDGA that recognizes the language
  $\dL[occur]{\a,\b,\c}$ of all $\{\a,\b,\c\}$-labeled graphs in which
  each of the three node labels occurs at least once. However,
  projection under the mapping $h\colon \{\a,\b,\c\}→\{\blank\}$, with
  $h(\a)=h(\b)=h(\c)=\blank$, yields the graph language
  $h(\dL[occur]{\a,\b,\c})=\dL[order]{≥3}$, which is not
  DDGA-recognizable (see the proof of \cref{lem:ddga<ndga}).
\end{proof}

\subsection{Equivalence of ADGAs and MSO-Logic} \label{sec:adga=mso}
\begin{theorem}[$\LL_\ADGA=\LL_\MSO$] \label{thm:adga=mso}
  A graph language is ADGA-recognizable \Iff it is
  MSO-definable. There are effective translations in both directions.
\end{theorem}

\begin{proofsketch}~
  \newcommand{\Fsucc}[1]{\text{$\varphi_{#1 \vphantom{i}}^{\hspace{.05ex}\textnormal{succ\vphantom{g}}}$}\hspace{-.1ex}}
  \newcommand{\Fwin}[1]{\text{$\varphi_{#1 \vphantom{i}}^{\hspace{.05ex}\textnormal{win\vphantom{g}}}$}\hspace{-.1ex}}
  \begin{itemize}
  \item[($⇒$)] We start with the direction $\LL_\ADGA⊆\LL_\MSO$. Let
    $\A=⟨Σ,Γ,\Q,σ,δ,\F⟩$ be an ADGA of length $n$. Without loss of
    generality, we may assume that every configuration reachable by
    $\A$ has at least one successor configuration and that no
    permanent configuration is reachable in less than $n$ rounds. In
    order to encode the acceptance behaviour of $\A$ into an
    $\MSO(Σ,Γ)$-sentence $\phiA$, we take again the game-theoretic
    point of view\footnote{This characterization is heavily inspired
      by the work of Löding and Thomas in \cite{LT00}.} briefly
    mentioned in the proof sketch of \cref{lem:adga-closure}.  Given
    $\A$ and some $G_λ∈Σ^{\clouded{Γ}}$, we consider a game with two
    players: the \emph{automaton} (player~$\EE$) and the
    \emph{pathfinder} (player~$\AA$). This game is represented by a
    directed acyclic graph whose nodes are precisely the
    configurations reachable by $\A$ on $G_λ$. For any two
    \emph{nonpermanent} configurations $G_κ$ and $G_μ$, there is a
    directed edge from $G_κ$ to $G_μ$ \Iff
    $G_μ∈δ^\cloud(G_κ)$. Starting at the initial configuration
    $G_{σ∘λ}$, the two players move through the game together by
    following directed edges. If the current configuration is
    existential, then the automaton has to choose the next move, if it
    is universal, then the decision belongs to the pathfinder. This
    continues until some permanent configuration is reached. The
    automaton wins if that permanent configuration is accepting,
    whereas the pathfinder wins if it is rejecting. A player is said
    to have a \emph{winning strategy} if it can always win,
    independently of its opponent's moves. It is straightforward to
    prove that the automaton has a winning strategy \Iff $\A$ accepts
    $G_λ$. Our MSO-sentence $\phiA$ will express the existence of such
    a winning strategy, and thus be equivalent to $\A$.

    Within MSO-logic, we represent a path $π=G_{κ_0}\cdots
    G_{κ_n}\strut$ through the game by a sequence of families of set
    variables $\wh{X}_0,…,\wh{X}_n$, where $\wh{X}_0 = ⟨\,⟩$ and
    $\wh{X}_i = ⟨\lsymb{U_{\meta{i},\meta{q}}}⟩_{q∈Q}$, for
    $1≤i≤n$. The intention is that each set variable
    $\lsymb{U_{\meta{i},\meta{q}}}$ is interpreted as the set of nodes
    $v∈\VG$ for which $κ_i(v)=q$. (We do not need set variables to
    represent $G_{κ_0}$, since the players always start at $G_{σ∘λ}$.)

    Now, for every round $i$, we construct a formula
    $\Fwin{i}[\wh{X}_i]$ (i.e., with free variables in $\wh{X}_i$),
    which expresses that the automaton has a winning strategy in the
    subgame starting at the configuration $G_{κ_i}$ represented by
    $\wh{X}_i$. In case $G_{κ_i}$ is existential, this is true if the
    automaton has a winning strategy in some successor configuration
    of $G_{κ_i}$, whereas if $G_{κ_i}$ is universal, the automaton
    must have a winning strategy in all successor configurations of
    $G_{κ_i}$. This yields the following recursive definition for
    $0≤i≤n-1$:
    \begin{equation*}
      \Fwin{i}[\wh{X}_i] \;\coloneqq\;
      \begin{cases}
        \logic{\displaystyle\bigexists \meta{\wh{X}_{i+\one}}
          \biggl(\meta{\Fsucc{i+\one}[\wh{X}_i,\wh{X}_{i+\one}]} \;\,∧\;\,
          \meta{\Fwin{i+\one}[\wh{X}_{i+\one}]} \biggr)}
        & \text{\parbox[c]{14ex}{if level $i$ of $\A$ is existential,}} \\[2.5ex]
        \logic{\displaystyle\bigforall \meta{\wh{X}_{i+\one}}
          \biggl(\meta{\Fsucc{i+\one}[\wh{X}_i,\wh{X}_{i+\one}]} \:⇒\:
          \meta{\Fwin{i+\one}[\wh{X}_{i+\one}]} \biggr)} 
        & \text{\parbox[c]{14ex}{if level $i$ of $\A$ is universal.}}
      \end{cases}
    \end{equation*}
    Here, $\Fsucc{i+1}[\wh{X}_i,\wh{X}_{i+1}]$ is an FO-formula
    expressing that $\wh{X}_i$ and $\wh{X}_{i+1}$ represent two
    configurations $G_{κ_i}$ and $G_{κ_{i+1}}$ such that
    $G_{κ_{i+1}}∈δ^\cloud(G_{κ_i})$. As our recursion base, we can
    easily construct a formula $\Fwin{n}[\wh{X}_n]$ that is satisfied
    \Iff $\wh{X}_n$ represents an accepting configuration of $\A$.

    The desired MSO-sentence is $\phiA \coloneqq\, \Fwin{0}[\wh{X}_0]
    = \Fwin{0}[\:]$.

  \item[($⇐$)] For the direction $\LL_\ADGA⊇\LL_\MSO$, we can proceed
    by induction on the structure of an $\MSO(Σ,Γ)$-formula $φ$. In
    order to deal with free occurrences of variables, we encode
    variable assignments into node labels. For $G_λ∈Σ^{\clouded{Γ}}$
    and $α\colon \free(φ)→\VG ∪ 2^\VG$, we represent $⟨G_λ,α⟩$ as the
    labeled graph $G_{λ×α^{-1}}$ whose labeling $λ\!×\!α^{-1}$ assigns
    to each node $v∈\VG$ the tuple $\bigl\langle λ(v),\:α^{-1}(v)
    \bigr\rangle$, where $α^{-1}(v)$ is the set of all variables in
    $\free(φ)$ to which $α$ assigns either $v$ or a set containing
    $v$. We now inductively construct an ADGA
    $\A_φ=⟨Σ\!×\!2^{\free(φ)},Γ,\Q,σ,δ,\F⟩$ such that
    \begin{equation*}
      G_{λ×α^{-1}}∈\L(\A_φ) \quad \text{\Iff} \quad\! ⟨G_λ,α⟩ ⊨ φ.
    \end{equation*}
    \newpage
    \begin{itemize}
    \item[(\texttt{BC})] Let $b∈Σ$,\: $τ∈Γ$,\: $x,y∈\Vnode$ and
      $X∈\Vset$.

      If $φ$ is one of the atomic formulas
      \noheight{$\logic{\lab{\meta{b}}\meta{x}}$},\:
      $\logic{\meta{x}=\meta{y}}$ or $\logic{\meta{x}∈\meta{X}}$,
      then, in $\A_φ$, each node simply checks that its own label
      $⟨a,M⟩ ∈ Σ×2^{\free(φ)}$ satisfies the condition specified in
      $φ$ (which, in particular, is the case if $x,y∉M$). Since this
      can be directly encoded into the initialization function $σ$,
      the ADGA has length $0$. It accepts the input graph \Iff every
      node reports that its label satisfies the condition.

      The case $φ=\logic{\meta{x}\xarr{\meta{τ}}\meta{y}}$ is very
      similar, but $\A_φ$ needs one communication round, after which
      the node assigned to $y$ can check whether it has received a
      message through a $τ$-edge from the node assigned to
      $x$. Accordingly, $\A_φ$ has length $1$.

    \item[(\texttt{IS})] In case $φ$ is a composed formula, we can
      obtain $\A_φ$ by means of the constructions outlined in the
      proof sketch of \cref{lem:adga-closure} (closure properties of
      $\LL_\ADGA$). Let $ψ$ and $ψ'$ be $\MSO(Σ,Γ)$-formulas with
      equivalent ADGAs $\A_ψ$ and $\A_{ψ'}$, respectively.

      If $φ=\logic{¬\meta{ψ}}$, it suffices to define $\A_φ=\cA_ψ$.
      Similarly, if $φ=\logic{\meta{ψ}∨\meta{ψ'}}$, we get $\A_φ$ by
      applying the union construction on $\A_ψ$ and $\A_{ψ'}$. (In
      general, we first have to extend $\A_ψ$ and $\A_{ψ'}$ such that
      they both operate on the same node alphabet
      $Σ\!×\!2^{\free(ψ)\,∪\,\free(ψ')}$.)

      Existential quantification can be handled by node projection. If
      $φ=\logic{∃\meta{X}(\meta{ψ})}$, with $X∈\Vset$, we construct
      $\A_φ$ by applying the projection construction on $\A_ψ$, using
      the mapping $h\colon Σ×2^{\free(ψ)} →
      Σ×2^{\free(φ)\:\!\setminus\:\!\{X\}}$ that deletes the set
      variable $X$ from every label. An analogous approach can be used
      if $φ=\logic{∃\meta{x}(\meta{ψ})}$, with $x∈\Vnode$. The only
      difference is that, instead of applying the projection
      construction directly on $\A_ψ$, we apply it on a variant
      $\A'_ψ$ that operates just like $\A_ψ$, but additionally checks
      that precisely one node in the input graph is assigned to the
      variable $x$.
      \qedhere
    \end{itemize}
  \end{itemize}
\end{proofsketch}

From \cref{thm:adga=mso} we can immediately infer that it is
undecidable whether the graph language recognized by some arbitrary
ADGA is empty. Otherwise, we could decide the satisfiability problem
of MSO-logic on graphs, which is known to be undecidable (a direct
consequence of Trakhtenbrot's Theorem, see, e.g.,
\cite[Thm~9.2]{Lib04}).

\begin{corollary}[Emptiness Problem of ADGAs] \label{cor:adga-emptiness}
  The emptiness problem of ADGAs is undecidable.
\end{corollary}

\subsection{Emptiness Problem of NDGAs} \label{sec:ndga-emptiness}
At the cost of reduced expressive power, we can also obtain a positive
decidability result.

\begin{lemma}[Emptiness Problem of NDGAs] \label{lem:ndga-emptiness}
  The emptiness problem of NDGAs is decidable in doubly-exponential
  time. More precisely, for every NDGA $\A=⟨Σ,Γ,\Q,\ab σ,δ,\F⟩$,
  whether its recognized graph language $\L(\A)$ is empty or not can
  be decided in time $2^k$\!, where $k∈\O\bigl(\:\!
  \card{Γ}·\card{Q}^{4\len(\A)}·\len(\A) \bigr)$.

  Furthermore, whether or not $\L(\A)$ contains any \emph{connected,
    undirected} graph can be decided in time $2^{2^{k'}}$\!\!, where
  $k'∈\O\bigl(\:\!  \card{Γ}·\card{Q}·\len(\A) \bigr)$.
\end{lemma}

\begin{proofsketch}
  Let $G_λ∈Σ^{\clouded{Γ}}$. Since NDGAs cannot perform universal
  branching, we can consider any run of $\A$ on $G_λ$ as a sequence of
  configurations $R=G_{κ_0}\!\cdots G_{κ_n}$, with $n≤\len(\A)$. In
  $R$, each node of $G$ traverses one of at most
  $\card{Q}^{\len(\A)+1}$ possible sequences of states. Now, assume
  that $G$ has more than $\card{Q}^{\len(\A)+1}$ nodes. Then, by the
  Pigeonhole Principle, there must be two distinct nodes $v,v'∈\VG$
  that traverse the same sequence of states in $R$. We construct a
  smaller graph $G'$ by removing $v'$ from $G$, together with its
  adjacent edges, and adding directed edges from $v$ to all of the
  former \emph{outgoing} neighbors of $v'$. If all the nodes in $G'$
  maintain their nondeterministic choices from $R$, none of them will
  notice that $v'$ is missing, and consequently they all behave just
  as in $R$. The resulting run $R'$ on $G'$ is accepting \Iff $R$ is
  accepting.

  Applying this argument recursively, we conclude that if $\L(\A)$ is
  not empty, then it must contain some labeled graph that has at most
  $\card{Q}^{\len(\A)+1}$ nodes. Hence, the emptiness problem is
  decidable because the search space is finite. The time complexity
  indicated above corresponds to the naive approach of checking every
  (directed) graph with at most $\card{Q}^{\len(\A)+1}$ nodes.

  If we are only interested in (connected) undirected graphs, the
  reasoning is very similar, but we have to require a larger minimum
  number of nodes in order to be able to remove some node without
  influencing the behavior of the others. In a graph $G$ with more
  than
  \mbox{$\bigl(\card{Q}·2^{\card{Γ}·\card{Q}}\bigr)^{\len(\A)+1}$}
  nodes, there must be two distinct nodes $v,v'∈\VG$ that, in addition
  to traversing the same sequence of states, also receive the same
  family of sets of states from their neighborhood in every
  round. Observe that the automaton will not notice if we merge $v$
  and $v'$. The rest of the argument is analogous to the previous
  scenario.
\end{proofsketch}

\subsection{Summary and Discussion}
We have introduced ADGAs, which are probably the first graph automata
in the literature to be equivalent to MSO-logic on graphs. However,
their expressive power results mainly from the use of alternation: we
have seen that the deterministic, nondeterministic and alternating
variants form a strict hierarchy, i.e.,
\begin{equation*}
  \LL_\DDGA \hyperref[lem:ddga<ndga]{⊂} \LL_\NDGA 
  \hyperref[lem:ndga<adga]{⊂} \LL_\ADGA.
\end{equation*} 
The corresponding closure and decidability properties are summarized
in \cref{tab:closure-decidability}.

\begin{table}[h!]
  \alignpic
  \begin{tabular}{lccccc}
    \toprule
         & \multicolumn{4}{c}{Closure Properties} & Decidability \\
    \cmidrule(rl){2-5} \cmidrule(rl){6-6}
         & Complement & Union & Intersection & Projection & Emptiness \\
    \addlinespace
    ADGA & \hyperref[lem:adga-closure]{\cmark} & \hyperref[lem:adga-closure]{\cmark}
         & \hyperref[lem:adga-closure]{\cmark} & \hyperref[lem:adga-closure]{\cmark} 
         & \hyperref[cor:adga-emptiness]{\xmark} \\
    \addlinespace
    NDGA & \hyperref[lem:ndga-closure]{\xmark} & \hyperref[lem:ndga-closure]{\cmark}
         & \hyperref[lem:ndga-closure]{\cmark} & \hyperref[lem:ndga-closure]{\cmark}
         & \hyperref[lem:ndga-emptiness]{\cmark} \\
    \addlinespace
    DDGA & \hyperref[lem:ddga-closure]{\cmark} & \hyperref[lem:ddga-closure]{\cmark}
         & \hyperref[lem:ddga-closure]{\cmark} & \hyperref[lem:ddga-closure]{\xmark}
         & \hyperref[lem:ndga-emptiness]{\cmark} \\
    \bottomrule
  \end{tabular}
  \caption{Closure and decidability properties of alternating,
    nondeterministic, and deterministic DGAs.}
  \label{tab:closure-decidability}
\end{table}

On an intuitive level, this hierarchy and these closure properties do
not seem very surprising. One might even ask: \emph{are ADGAs just
  another syntax for MSO-logic?} Indeed, universal branchings
correspond to universal quantification, and nondeterministic choices
to existential quantification. By disallowing universal set
quantification in MSO-logic we obtain EMSO-logic, and further
disallowing existential set quantification yields
FO-logic. Analogously to DGAs, the classes of graph languages
definable in these logics form a strict hierarchy, i.e.,
\begin{equation*}
  \LL_\FO ⊂ \LL_\EMSO ⊂ \LL_\MSO.
\end{equation*} 
Furthermore, the closure properties of $\LL_\EMSO$ and $\LL_\FO$
coincide with those of $\LL_\NDGA$ and $\LL_\DDGA$,
respectively. Given that $\LL_\ADGA$ and $\LL_\MSO$ are equal, one
might therefore expect that the analogous equalities hold for the
weaker classes. However, as already hinted by the positive
decidability properties in \cref{tab:closure-decidability}, this is
not the case.  The actual relationships between the different classes
of graph languages are depicted in \cref{fig:venn_diagram}. A glance
at this Venn diagram suggests that ADGAs are not simply a one-to-one
reproduction of MSO-logic.

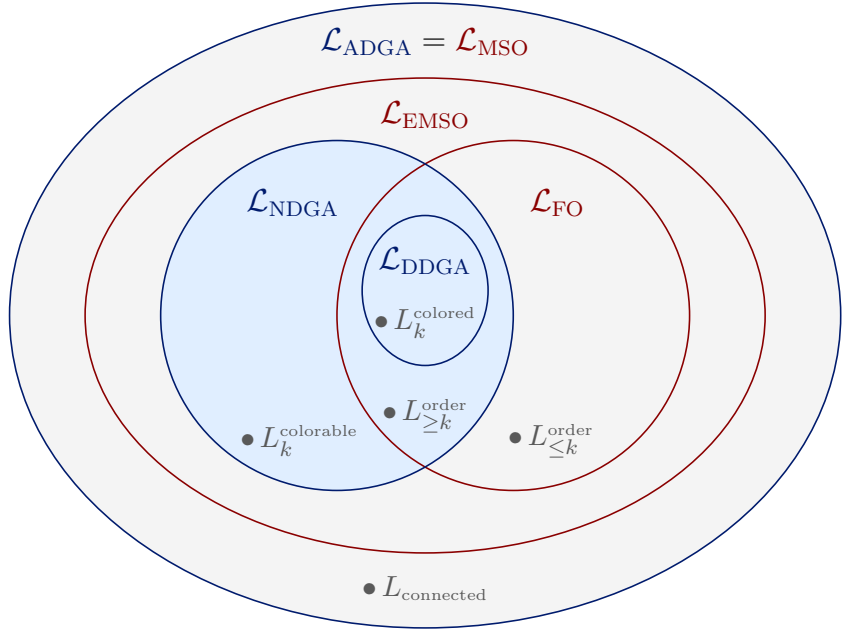
\begin{figure}
  \alignpic
  \input{fig/venn_diagram.tikz}
  \caption{Venn diagram relating the classes of graph languages
    recognizable by our three flavors of DGAs to those definable in
    MSO-, EMSO- and FO-logic.}
  \label{fig:venn_diagram}
\end{figure}

\begin{proof}[Justification of \cref{fig:venn_diagram}]
  Fagin has shown in \cite{Fag75} that the language $\sL{connected}$
  of all (weakly) connected graphs separates $\LL_\EMSO$ from
  $\LL_\MSO$. (Since non-connectivity is EMSO-definable, this also
  implies that $\LL_\EMSO$ is not closed under complementation.) The
  inclusion $\LL_\NDGA⊆\LL_\EMSO$ holds because we can encode every
  NDGA into an EMSO-sentence, using the same construction as in the
  proof sketch of \cref{thm:adga=mso}. It is also easy to see that we
  do not need any set quantifiers to encode DDGAs, hence
  $\LL_\DDGA⊆\LL_\FO$. In the following, let $k,k'≥2$. The
  incomparability of $\LL_\NDGA$ and $\LL_\FO$ is witnessed by the
  language $\dL[colorable]{k}$ of $k$-colorable graphs, which lies
  within $\LL_\NDGA$ (see Example~\ref{ex:ADGA_3_colorable}) but
  outside of $\LL_\FO$ (see, e.g., \cite{Lib04}), and the language
  $\dL[order]{≤k}$ of graphs with at most $k$ nodes, which lies
  outside of $\LL_\NDGA$ (see the proof of \cref{lem:ndga<adga}) but
  obviously within $\LL_\FO$. Considering the union language
  $\dL[colorable]{k}∪\dL[order]{≤k'}$ also tells us that the inclusion
  of $\LL_\NDGA∪\LL_\FO$ in $\LL_\EMSO$ is strict. Finally, the
  language $\dL[order]{≥k}$ of graphs with at least $k$ nodes
  separates $\LL_\DDGA$ from $\LL_\NDGA∩\LL_\FO$ (see the proof of
  \cref{lem:ddga<ndga}). A simple example of a language that lies
  within $\LL_\DDGA$ is the set $\dL[colored]{k}$ of $Σ$-labeled
  graphs whose labelings are valid $k$-colorings, with $\card{Σ}=k$.
\end{proof}

As of the time of writing this paper, no new results on $\LL_\MSO$
have been inferred from the alternative characterization through
ADGAs. On the other hand, the notion of NDGA contributes to the
general observation, mentioned in \cref{sec:introduction}, that many
characterizations of regularity, which are equivalent on words and
trees, drift apart on graphs. To see this, consider NDGAs whose input
is restricted to those $Σ$-labeled $Γ$-graphs that represent words or
trees over the alphabet $Σ$. For words, $Γ=\{\blank\}$ and edges
simply go from one position to the next, whereas for ordered trees of
arity $k$, we set $Γ=\{1,…,k\}$ and require edge relations such that
\noheight{$u\xarr{i}v$} \Iff $u$ is the $i$-th child of $v$. Observe
that we can easily simulate any word or tree automaton by an NDGA of
length $2$: guess a run of the automaton in the first round (each node
nondeterministically chooses some state), then check whether it is a
valid accepting run in the second round (transitions are verified
locally, and acceptance is determined by the unique sink). This
implies that the classes of NDGA-recognizable and MSO-definable
languages collapse on words and trees, and hence that NDGAs recognize
precisely the regular languages on those restricted structures.

The fact that the emptiness problem of NDGAs is decidable on graphs
seems noteworthy for several reasons:
\begin{itemize}[topsep=1ex,itemsep=0ex]
\item It can be seen as an extension to graphs of the corresponding
  decidability results for finite automata on words and trees, since,
  by the above remark, the emptiness problems of these automata
  correspond precisely to those of NDGAs restricted to words and
  trees, respectively.
\item It might lead to the discovery of new decidable logics on
  graphs: a logic effectively equivalent to NDGAs would have a
  decidable satisfiability problem, and a logic effectively equivalent
  to DDGAs would additionally have a decidable validity problem. This
  could be interesting when contrasted with Trakhtenbrot's Theorem,
  which states that these problems are undecidable for FO-logic, and a
  fortiori for (E)MSO-logic (see, e.g., \cite[Thm~9.2]{Lib04}).
\item It implies that the language inclusion problem of DDGAs is also
  decidable: given two DDGAs $\A_1$ and $\A_2$, we can decide whether
  $\L(\A_1)⊆\L(\A_2)$ by first applying the intersection construction
  on $\A_1$ and a complement of $\A_2$, and then deciding emptiness
  for the resulting automaton. (This does not extend to NDGAs, since
  they do not satisfy closure under complementation.) The verification
  method presented in the next section is based on such an inclusion
  test.
\end{itemize}

%% file: fig/ADGA_3_colorable.tikz
\begin{tikzpicture}[automaton, half row sep]
  \matrix[states] {
        & \node[existential] (q_p) {$q_\pik$}; &[6ex] \\
        &     & \node[permanent] (q_yes) {$\q{yes}$}; \\
    \node[initial,existential] (q_ini) {$\q{ini}$}; & \node[existential] (q_h) {$q_\herz$}; \\
        &     & \node[permanent] (q_no) {$\q{no}$}; \\
        & \node[existential] (q_k) {$q_\kreuz$}; \\
  };
  \path[use as bounding box]
        (q_ini)   edge (q_p)
                  edge (q_h)
                  edge (q_k)
        (q_p)     edge[bend left=25] node[above=-.3ex,xshift=.4ex] {$\xnni q_\pik$} (q_yes)
                  edge[bend left=15] node[above] {$\xni q_\pik$} (q_no)
        (q_h.10)  edge[bend right=5] node[above left=-.5ex,xshift=-1.5ex] {$\xnni q_\herz$} (q_yes)
        (q_h.350) edge[bend left=5] node[below left=-.5ex,xshift=-1.5ex] {$\xni q_\herz$} (q_no)
        (q_k)     edge[bend right=15] node[below=.3ex] {$\xnni q_\kreuz$} (q_yes)
                  edge[bend right=25] node[below,xshift=.4ex] {$\xni q_\kreuz$} (q_no);
  \matrix[accepting sets] {
       \\
    \x \\
       \\
       \\
       \\
  };
  \DrawColumnBackground{2}{4}{1}
\end{tikzpicture}

%% file: fig/ADGA_not_3_colorable.tikz
\begin{tikzpicture}[automaton, half row sep]
  \matrix[states] {
        & \node[universal] (q_p) {$q_\pik$}; &[6ex] \\
        &     & \node[permanent] (q_yes) {$\q{yes}$}; \\
    \node[initial,universal] (q_ini) {$\q{ini}$}; & \node[universal] (q_h) {$q_\herz$}; \\
        &     & \node[permanent] (q_no) {$\q{no}$}; \\
        & \node[universal] (q_k) {$q_\kreuz$}; \\
        & \\
  };
  \path[use as bounding box]
        (q_ini.20) edge (q_p)
        (q_ini)    edge (q_h)
                   edge (q_k)
        (q_p)      edge[bend left=25] node[above=-.3ex,xshift=.1ex] {$\xnni q_\pik$} (q_yes)
                   edge[bend left=15] node[above] {$\xni q_\pik$} (q_no)
        (q_h.10)   edge[bend right=5] node[above left=-.5ex,xshift=-1.5ex] {$\xnni q_\herz$} (q_yes)
        (q_h.350)  edge[bend left=5] node[below left=-.5ex,xshift=-1.3ex] {$\xni q_\herz$} (q_no)
        (q_k)      edge[bend right=15] node[below=.4ex,xshift=-.2ex] {$\xnni q_\kreuz$} (q_yes)
                   edge[bend right=25] node[below,xshift=.5ex] {$\xni q_\kreuz$} (q_no);
  \matrix[accepting sets] {
            \\
       & \x \\
            \\
    \x & \x \\
            \\
            \\
  };
  \DrawColumnBackground{2}{4}{2}
\end{tikzpicture}
\vspace{-2ex}

%% file: fig/ADGA_concentric_circles.tikz
\begin{tikzpicture}[automaton, half row sep]
  \matrix[states] {
        &[4ex]&[7ex] \node[permanent] (q_ap) {$\qap$}; \\
    \node[existential] (q_a) {$\qa$}; & \node[universal] (q_a2) {$\qaprime$};  \\
        &     & \node[permanent] (q_ah) {$\qah$}; \\
        & \node[universal] (q_bkr) {$\qbkr$}; \\
    \node[existential] (q_b) {$\qb$}; &     & \node[permanent] (q_yes) {$\q{yes}$}; \\
        & \node[universal] (q_bk) {$\qbk$}; \\
        \\
    \node[existential] (q_c) {$\qc$}; &     & \node[permanent] (q_no) {$\q{no}$}; \\
  };
  \matrix[symbols] { 
        \\
    \node (a) {$\a$}; \\
        \\
        \\
    \node (b) {$\b$}; \\
        \\
        \\
    \node (c) {$\c$}; \\
  };
  \matrix[accepting sets] {
    \x &    \\
       &    \\
       & \x \\
       &    \\
    \x & \x \\
       &    \\
       &    \\
       &    \\
  };
  \DrawColumnBackground{1}{8}{2}
  \path (a)        edge (q_a)
        (b)        edge (q_b)
        (c)        edge (q_c)
        (q_a)      edge (q_a2)
        (q_b.25)   edge node[above] {$\xnni \qb$} (q_bkr)
        (q_b.0)    edge node[above,xshift=.2ex] {$\xnni \qb$} (q_bk)
        (q_b.330)  edge[bend right=20] node[above=4.2ex,xshift=-9.2ex] {$\xni \qb$} (q_no) 
        (q_c.0)    edge[bend right=25] node[above=-.5ex,xshift=-12ex] {$\xnni \qc ∧ \xnni \qa$} (q_yes)
        (q_c.330)  edge[bend right=22] node[above=.2ex] {$\xni \qc ∨ \xni \qa$} (q_no)
        (q_a2.22)  edge[bend left=8] node[above=.3ex,xshift=.1ex] {$\xeq \{\qbkr,\qbk\}$} (q_ap)
        (q_a2)     edge node[above=.1ex,xshift=.4ex] {$\xeq \{\qbkr,\qbk\}$} (q_ah)
        (q_a2.320) edge[bend right=9] node[above=4.5ex,xshift=1.3ex] {$\xneq \{\qbkr,\qbk\}$} (q_no)
        (q_bkr)    edge (q_yes)
        (q_bk)     edge (q_yes);
\end{tikzpicture}

%% file: fig/graph_labeled_pentagon.tikz
%
\pentagraphPic{input graph}{\lnodedistIG}{$\a$}{$\b$}{$\c$}{$\b$}{$\b$}{$\c$}

%% file: fig/run_accepting.tikz
\begin{tikzpicture}[run or game]
  \matrix {
      & & \node[config,pacc] (c3a)
           {\pentagraphPic{configuration}{\lnodedistC}{$\qap$}{$\q{yes}$}{$\q{yes}$}{$\q{yes}$}{$\q{yes}$}{$\q{yes}$}}; \\
    \node[config,exis] (c1)
     {\pentagraphPic{configuration}{\lnodedistC}{$\qa$}{$\qb$}{$\qc$}{$\qb$}{$\qb$}{$\qc$}}; 
      & \node[config,univ] (c2)
         {\pentagraphPic{configuration}{\lnodedistC}{$\qaprime$}{$\qbkr$}{$\q{yes}$}{$\qbkr$}{$\qbk$}{$\q{yes}$}}; \\
      & & \node[config,pacc] (c3b)
           {\pentagraphPic{configuration}{\lnodedistC}{$\qah$}{$\q{yes}$}{$\q{yes}$}{$\q{yes}$}{$\q{yes}$}{$\q{yes}$}}; \\
  };
  \path (c1) edge (c2)
        (c2) edge (c3a)
             edge (c3b);
\end{tikzpicture}

%% file: fig/venn_diagram.tikz
\begin{tikzpicture}[semithick]
  \tikzstyle{classL} = [font=\large]
  \tikzstyle{exampleL} = [color=darkgray]

  \def\Ladga{(0,0) ellipse (33ex and 25ex)}
  \def\Lemso{(0,0) ellipse (27ex and 19ex)}
  \def\Lndga{(180:7ex) circle (14ex)}
  \def\Lfo{(0:7ex) circle (14ex)}
  \def\Lddga{(90:2ex) ellipse (5ex and 6ex)}

  \draw[darkblue] \Ladga node[above=20ex,classL] {\color{black}${\color{darkblue}\LL_\ADGA}={\color{darkred}\LL_\MSO}$}
      node[below=20ex,exampleL] {$•\,\sL{connected}$};
  \draw[darkred] \Lemso node[above=14.2ex,classL] {$\LL_\EMSO$};
  \draw[darkred] \Lfo node[xshift=3.5ex,yshift=9ex,classL] {$\LL_\FO$}
      node[xshift=3ex,yshift=-10ex,exampleL] {$•\,\dL[order]{≤k}$};
  \draw[darkblue] \Lndga node[xshift=-3.5ex,yshift=9ex,classL] {$\LL_\NDGA$}
      node[xshift=-3ex,yshift=-10ex,exampleL] {$•\,\dL[colorable]{k}$};
  \draw[darkblue] \Lddga node[above=.5ex,classL] {$\LL_\DDGA$}
      node[below=.5ex,exampleL] {$•\,\dL[colored]{k}$};
  \node[exampleL] at (-90:8ex) {$•\,\dL[order]{≥k}$};

  \begin{scope}[on background layer]
    \fill[llightgray!70!white] \Ladga;
    \fill[lightblue] \Lndga;
  \end{scope}
\end{tikzpicture}

%% file: sec/DA_verification.tex
\section{Verification of Distributed Algorithms} \label{sec:DA_verification}
The notion of graph automaton might have an application in formal
verification of synchronous distributed algorithms. In this section,
we consider a very simple toy example of such an algorithm, and
suggest a mechanical verification technique based on DDGAs for proving
partial correctness, using Floyd-Hoare logic. So far, our approach
only works for an extremely restricted class of synchronous
algorithms. However, since the method does not intrinsically depend
upon a particular automaton model, it is possible, in principle, to
extend it by replacing DDGAs with a more powerful class of graph
automata. In this regard, the following method should be considered as
an illustration of a concept, rather than a “ready-to-use” solution.

\subsection{Distributed Programming Language}
As mentioned by Konnov et al.\ in \cite{KVW12}, one of the major
obstacles in formal verification of distributed algorithms is the lack
of a versatile formal language to specify such algorithms. They refer
to it as the \emph{formalization problem}. Indeed, most of the
distributed algorithms found in the literature are given as
pseudocode, since implementation details are generally not the main
concern.

Here, we restrict ourselves to a very weak class of synchronous
algorithms for which the formalization problem can be easily solved.
(This is not, by any means, an attempt at a general solution.) We
design our programming language in such a way that individual
synchronous rounds can be simulated by a DDGA. In particular, this
means that we only consider algorithms where
\begin{itemize}[topsep=1ex,itemsep=0ex]
\item the nodes have a finite state space,
\item they send the same message to all of their neighbors, and
\item they only receive a set containing all the messages sent by
  their neighbors.
\end{itemize}
Furthermore, in contrast to classical distributed algorithms, we
express loops from the global point of view of a controller that can
see the states of all the nodes at once. This will allow us to partly
reason about distributed algorithms as if they were ordinary
sequential algorithms and employ the inference rules from Floyd-Hoare
logic. Obviously, the presence of a global controller introduces some
additional expressive power which is not available in a purely
distributed setting. We shall make use of it to model supplementary
knowledge that the nodes might have about the graph. As a matter of
fact, it is often assumed in distributed computing that the nodes know
properties such as the order of the graph or its diameter.

We shall assume that our algorithms always run on \emph{connected},
\emph{undirected} graphs. The former property is usually required in
distributed computing because nodes in separate connected components
are unable to communicate with each other, which de facto means that
any distributed algorithm is executed separately in each connected
component. Assuming that graphs are undirected is also very common,
and generally leads to simpler algorithms. In order to restrict the
possible input graphs of our automata accordingly, for every DDGA
$\A$, we denote by \defd{$\L_\UnCo(\A)$} the set of connected,
undirected labeled graphs that are accepted by $\A$. For the remainder
of this section, since no confusion with $\L(\A)$ can arise, we will
also refer to $\L_\UnCo(\A)$ as the graph language \defd{recognized}
by $\A$, and say that it is \defd{DDGA-recognizable}.

We now semi-formally specify the syntax and semantics of our
distributed programming language. Any considered distributed algorithm
operates on a set of \defd{variables} $\defd{\Var} = \{x_1,…,x_k\}$
ranging over \defd{values} from some \emph{finite} domain
$\defd{\Val}$. Each node $v$ of the input graph has its own private
copies of these variables, which are denoted by $v.x_1,…,v.x_k$ and
are referred to as $v$'s \defd{member variables}. The \defd{global
  state} of a graph is given by a valuation of the member variables of
all of its nodes. Formally, any $(\Val^\Var)$-labeled graph $G_λ$ is a
global state of the graph $G$, i.e., we label the nodes of $G$ with
functions from $\Var$ to $\Val$.

The commands executed locally by a node $v$ can contain expressions
evaluated over $\Val$. The syntax of these expressions is given by
\begin{equation*}
  \NT{e} \Coloneqq v.x \bigm| f(\NT{e},…,\NT{e}) \bigm| f(M,\NT{e},…,\NT{e}),
\end{equation*}
where $x∈\Var$, $M$ is a special set variable that does not contribute
to the global state, and $f$ represents some function from
$(\Val×\cdots×\Val)$ or $(2^\Val×\Val×\cdots×\Val)$ into $\Val$.

Similarly, we allow Boolean expressions of the form
\begin{equation*}
  \NT{b} \Coloneqq f(\NT{e},…,\NT{e}) \bigm| f(M,\NT{e},…,\NT{e}),
\end{equation*}
where the function associated with $f$ maps into the Boolean domain.

As elementary \defd{local commands}, any node $v$ can either do
nothing (\Skip) or assign a new value to one of its member
variables. Furthermore, local commands can be composed sequentially
and executed conditionally. The corresponding syntax is given by
\begin{equation*}
  \NT{C} \Coloneqq \text{\Skip;}
  \bigm| \text{\Assign{$v.x$}{$\NT{e}$};}
  \bigm| \text{$\NT{C}$ $\NT{C}$}
  \bigm| \text{\xIf $\NT{b}$ \xThen $\NT{C}$ \xElse $\NT{C}$ \xEnd}.
\end{equation*}

A \defd{local command block} executed by $v$ consists of a sequence of
local commands. It can optionally be preceded by a synchronous message
exchange, where $v$ sends the value of one of its member variables
$v.x$ to all of its neighbors, and in return receives a set of values
which is assigned to the dedicated set variable $M$. The only purpose
of $M$ is to access the set of incoming messages, and its scope is
restricted to the local command block.
\begin{equation*}
  \NT{D} \Coloneqq \text{$\NT{C}$}
  \bigm| \text{\SendReceive{$v.x$}{$M$};\, $\NT{C}$}
\end{equation*}

Next, we switch to a global perspective where we can control which
local command blocks are executed by the nodes. As an elementary
\defd{global command}, we can tell all the nodes to execute a
particular local command block synchronously in parallel. This
corresponds to a single synchronous round of a distributed
algorithm. To express more complex algorithms, global commands can be
composed sequentially and executed in loops. The syntax is of the form
\begin{equation*}
  \NT{K} \Coloneqq \text{\xEach $v$ \xDoes $\NT{D}$ \xEnd}
  \bigm| \text{$\NT{K}$ $\NT{K}$}
  \bigm| \text{\xWhile $ξ$ \xDo $\NT{K}$ \xEnd},
\end{equation*}
where $ξ$ is a textual representation of a condition on the global
state of the graph that can be expressed as a DDGA-recognizable graph
language. We compactly represent such conditions by FO-formulas over
the node alphabet $\Val^\Var$\!, and refer to them as
\defd{DDGA-recognizable assertions}. Since DDGA-recognizable graph
languages are closed under Boolean set operations, there are no
restrictions on combining these formulas using the usual propositional
connectives. Also, for convenience, we shall use many syntactic
abbreviations whose meaning should be clear. For instance, by the
abbreviation \noheight{$\logic{\meta{\lsymb{v}.x_i=c}}$} we mean the
disjunction of all the formulas \noheight{$\logic{\lab{\meta{a}}v}$}
such that $a∈\Val^\Var$ and $a(x_i)=c$.

\bigskip 

Let us consider the \textsc{FloodMax} algorithm as an example of a
simple distributed algorithm that can be expressed in the programming
language we just defined.
\begin{example}[{\sc FloodMax} Algorithm] \label{ex:floodmax}
  Initially, each node is given a number from some finite domain. The
  task for the nodes is to compute the maximum number $m$ present in
  the graph in such a way that, once the algorithm has terminated, all
  of them know $m$. An algorithm solving this problem can be used, for
  instance, to solve the leader election problem (see, e.g.,
  \cite[Sec~4.1]{Lyn96}).

  If we assume that the nodes know the diameter $d$ of the graph, a
  simple approach is as follows: In each synchronous round, each node
  sends the maximum number it has seen so far (initially its own) to all
  of its neighbors. After $d$ rounds, we are guaranteed that every node
  has received the global maximum.

  In order to formalize this algorithm in our framework, we must somehow
  exploit the power of the global controller to simulate the
  circumstance that the nodes know $d$. However, since the controller
  can only check DDGA-recognizable assertions, it has no way of knowing
  $d$ itself. Fortunately, this is not necessary. Counting the number of
  rounds up to $d$ is only a way of ensuring that enough time has passed
  for information to propagate between any two nodes. Alternatively, the
  algorithm could also terminate as soon as no node receives any new
  information. Although this condition cannot be checked in a purely
  distributed setting, it can be specified by a DDGA, and thus leads to
  a variant of the algorithm that we can formalize.

  A possible way of formalizing \textsc{FloodMax} can be seen in
  Algorithm~\ref{alg:floodmax}. This algorithm operates on the
  variables $m$ and $m_\old$, which take values from some finite set
  $I$ of nonnegative integers that contains $0$. Initially, for each
  node $v$, a number in $I$ is assigned to $v.m$ as input. Then, in
  each round, $v$ updates its member variables in such a way that
  $v.m$ holds the maximum value it has seen so far and $v.m_\old$
  holds the value of $v.m$ from the previous round (except for the
  first round, where it is set to $0$). The algorithm terminates as
  soon as for every node $v$ the member variables $v.m$ and $v.m_\old$
  have the same value, a property (here represented by the assertion
  \noheight{$\logic{∃v(\meta{\lsymb{v}.m ≠ \lsymb{v}.m_\old})}$}) that
  can be easily checked by a DDGA over the node alphabet
  $I^{\{m,\,m_\old\}}$.
\end{example}

\begin{algorithm}
  \caption{\textsc{FloodMax}}
  \label{alg:floodmax}
  \Each{$v$}{
    \Assign{$v.m_\old$}{$0$}\;
  }
  \While{$\logic{∃v(\meta{\lsymb{v}.m ≠ \lsymb{v}.m_\old})}$}{
    \Each{$v$}{
      \SendReceive{$v.m$}{$M$}\;
      \Assign{$v.m_\old$}{$v.m$}\;
      \Assign{$v.m_{\phantom{\old}}$}{$\max(M ∪ \{v.m\})$}\;
    }
  }
\end{algorithm}

\subsection{Verification Method}
Now that we have a formal language for representing certain
distributed algorithms, we can turn towards the verification method
mentioned earlier. The basic idea is to consider a synchronous
distributed algorithm as an ordinary sequential one, and treat each
round of that algorithm as an atomic operation on the global state of
the graph. Consequently, once we know how to derive a Hoare triple for
a single round, we can simply use classical Floyd-Hoare logic to prove
partial correctness of an entire algorithm.

In our framework, a synchronous round is represented by a global
command $K \,=\, \text{\xEach $v$ \xDoes $D$ \xEnd}$, where $D$ is
some local command block. We desire an inference rule that allows us
to derive a Hoare triple of the form $\hoare{φ}\,K\,\hoare{ψ}$, where
$φ$ and $ψ$ are DDGA-recognizable assertions on the global state of
the graph.

As hinted previously, we can now take advantage of the restrictions
that we have put on the considered algorithms in order to simulate $K$
by a DDGA. This allows us to construct an automaton for the
\defd{weakest precondition} under $K$ of any DDGA-recognizable
assertion. Given some DDGA $\A$ over the node alphabet $\Val^\Var$\!,
we define a DDGA \defd{$\Wp(K,\A)$} over the same alphabet such that
$\Wp(K,\A)$ first simulates the execution of $K$ on the input graph
and then operates like $\A$ on the resulting graph. In other words,
$\Wp(K,\A)$ is an automaton that accepts a $(\!\Val^\Var)$-labeled
graph $G_λ$ \Iff $\A$ accepts the labeled graph $G_{λ'}$ that one
obtains after running the global command $K$ on $G_λ$. Such an
automaton can be effectively constructed because $\Var$ and $\Val$ are
finite and the message exchange process in our algorithm framework is
the same as for DDGAs. We can easily obtain $\Wp(K,\A)$ by inserting
one additional level of states before the first level of $\A$.

This brings us to the desired inference rule for $K \,=\, \text{\xEach
  $v$ \xDoes $D$ \xEnd}$:
\begin{equation*}
  \tag{single round}
  \inferrule{\L_\UnCo(\A_φ)⊆\L_\UnCo\bigl(\Wp(K,\A_ψ)\bigr)}{\hoare{φ}\,K\,\hoare{ψ}},
\end{equation*}
where $\A_φ$ and $\A_ψ$ designate DDGAs that recognize the properties
specified by $φ$ and $ψ$, respectively. This rule reduces the problem
of deriving a Hoare triple for a single round to the language
inclusion problem of DDGAs. We know that the latter is decidable,
because it can, in turn, be reduced to the emptiness problem of DDGAs
on connected, undirected graphs, which we have shown to be decidable
in \cref{lem:ndga-emptiness}. (This second reduction also relies on
the fact that DDGA-recognizable graph languages are effectively closed
under complementation and intersection, which holds by
\cref{lem:ddga-closure}.)

For more complex global commands, we simply take over the inference
rules from classical Floyd-Hoare logic, i.e., for any global commands
$K$, $K_1$, $K_2$ and DDGA-recognizable assertions $φ$, $φ'$, $ψ$,
$ψ'$, $θ$, $ξ$, we have
\begin{equation*}
  \tag{sequence}
  \inferrule{\hoare{φ}\:K_1\:\hoare{θ}\\\hoare{θ}\:K_2\:\hoare{ψ}}{\hoare{φ}\:\text{$K_1$ $K_2$}\:\hoare{ψ}},
\end{equation*}
\begin{equation*}
  \tag{loop}
  \inferrule{\hoare{\logic{\meta{θ}∧\meta{ξ}}}\:K\:\hoare{θ}}{\hoare{θ}\:\text{\xWhile $ξ$ \xDo $K$ \xEnd}\:\hoare{\logic{\meta{θ}∧¬\:\!\meta{ξ}}}},
  \quad\text{and}
\end{equation*}
\vspace{.5ex}
\begin{equation*}
  \tag{strengthen/weaken}
  \inferrule{\L_\UnCo(\A_φ)⊆\L_\UnCo(\A_{φ'})\\\hoare{φ'}\,K\,\hoare{ψ'}\\\L_\UnCo(\A_{ψ'})⊆\L_\UnCo(\A_ψ)}{\hoare{φ}\,K\,\hoare{ψ}}.
\end{equation*}

\bigskip 

We can now use this method to verify the \textsc{FloodMax} algorithm
from \cref{ex:floodmax}.
\begin{example}[Verification of {\sc FloodMax}]
  An assertion-annotated version of the code is displayed in
  Algorithm~\ref{alg:floodmax_annot}. As is often the case in
  Floyd-Hoare proofs, verification is performed with the help of
  auxiliary variables, which may not be modified by the algorithm to
  be verified. We introduce an additional member variable $v.m_\ini$
  for each node $v$ in order to be able to refer to the initial value
  of $v.m$, and, accordingly, our precondition is
  $\logic{∀v(\meta{\lsymb{v}.m=\lsymb{v}.m_\ini})}$. The algorithm
  satisfies partial correctness if, after termination, every node $v$
  knows the maximum initial value present in the graph, i.e., the
  postcondition is 
  $\logic{ ∀v(\meta{\displaystyle\lsymb{v}.m=\max_{\lsymb{u}}\,
      \lsymb{u}.m_\ini})}$.

  The crucial part of the proof is finding a suitable loop invariant.
  Assertion~$θ$ in Algorithm~\ref{alg:floodmax_annot} turns out to be
  adequate. It states that for each node $v$, the currently largest
  known number $v.m$ is bounded from below by the largest values known
  to its neighbors in the previous round and by its own initial value
  $v.m_\ini$, and that there is some node in the graph that still
  retains its initial value. If we denote by $K$ the loop body
  (lines~\ref{lin:loop_body_start} to \ref{lin:loop_body_end}) and by
  $\A_θ$ a DDGA equivalent to $θ$, it is easy to see that
  $\L_\UnCo(\A_θ)⊆\L_\UnCo\bigl(\Wp(K,\A_θ)\bigr)$, and consequently
  we can use the “single round” rule to derive the Hoare triple
  $\hoare{θ}\,K\,\hoare{θ}$. Note that our invariant does not even
  rely on the exit condition $ξ$ of the loop. By the
  “strengthen/weaken” rule, we obtain
  \noheight{$\hoare{\logic{\meta{θ}∧\meta{ξ}}}\,K\,\hoare{θ}$}, and
  then the “loop” rule allows us to derive
  \noheight{$\hoare{θ}\:\text{\xWhile $ξ$ \xDo $K$
      \xEnd}\:\hoare{\logic{\meta{θ}∧¬\:\!\meta{ξ}}}$}. It is also
  easy to see that the loop invariant $θ$ follows from assertion~$φ$,
  and that \noheight{$\logic{\meta{θ}∧¬\:\!\meta{ξ}}$} implies
  assertion~$ψ$. Again, this can be expressed in terms of inclusions
  of DDGA-recognizable graph languages, and by the “strengthen/weaken”
  rule we get $\hoare{φ}\:\text{\xWhile $ξ$ \xDo $K$
    \xEnd}\:\hoare{ψ}$.

  By proceeding analogously for the initialization part
  (lines~\ref{lin:init_start} to \ref{lin:init_end}), and then applying
  the “sequence” and “strengthen/weaken” rules, we can formally derive
  the desired Hoare triple
  \begin{equation*}
    \hoare{\logic{∀v(\meta{\lsymb{v}.m=\lsymb{v}.m_\ini})}}
    \:\text{\textsc{FloodMax}}\:
    \hoare{\logic{∀v(\meta{\displaystyle\lsymb{v}.m=\max_{\lsymb{u}}\, \lsymb{u}.m_\ini})}}.
  \end{equation*}
  Note that the postcondition follows from assertion~$ψ$ because we
  only consider graphs that are undirected and connected. The first
  conjunct of $ψ$ implies that all the nodes $v$ have the same value
  for $v.m$, while the two remaining conjuncts ensure that this value
  is indeed the maximum initial value present in the graph.
\end{example}

\begin{algorithm}
  \caption{\textsc{FloodMax} with Floyd-Hoare Annotations}
  \label{alg:floodmax_annot}
  \Hoare{$\logic{∀v(\meta{\lsymb{v}.m=\lsymb{v}.m_\ini})}$}
  \Each{$v$}{ \label{lin:init_start}
    \Assign{$v.m_\old$}{$0$}\;
  } \label{lin:init_end}
  \Hoare[$φ$]{$\logic{∀v(\meta{\lsymb{v}.m=\lsymb{v}.m_\ini} \,∧\, \meta{\lsymb{v}.m_\old=\zero})}$}
  \While{\Tagged{$ξ$}{$\logic{∃v(\meta{\lsymb{v}.m ≠ \lsymb{v}.m_\old})}$}}{
    \Hoare[$θ$]{$\logic{∀u,v(u\!\arr\!v \,⇒\, \meta{\lsymb{u}.m_\old≤\lsymb{v}.m}) \,∧\, ∀v(\meta{\lsymb{v}.m≥\lsymb{v}.m_\ini}) \,∧\, ∃v(\meta{\lsymb{v}.m=\lsymb{v}.m_\ini})}$}
    \Each{$v$}{ \label{lin:loop_body_start}
      \SendReceive{$v.m$}{$M$}\;
      \Assign{$v.m_\old$}{$v.m$}\;
      \Assign{$v.m_{\phantom{\old}}$}{$\max(M ∪ \{v.m\})$}\;
    } \label{lin:loop_body_end}
  }
  \Hoare[$ψ$]{$\logic{∀u,v(u\!\arr\!v \,⇒\, \meta{\lsymb{u}.m≤\lsymb{v}.m}) \,∧\, ∀v(\meta{\lsymb{v}.m≥\lsymb{v}.m_\ini}) \,∧\, ∃v(\meta{\lsymb{v}.m=\lsymb{v}.m_\ini})}$}
  \Hoare{$\logic{∀v(\meta{\displaystyle\lsymb{v}.m=\max_{\lsymb{u}}\, \lsymb{u}.m_\ini})}$}
\end{algorithm}

\subsection{Prospects and Limitations}
As mentioned at the beginning of this section, our adaptation of
Floyd-Hoare logic is in principle not restricted to the toy language
presented here. By replacing DDGAs with a more powerful class of (not
necessarily finite-state) graph automata, we might directly obtain a
variant of the framework in which we could formalize and verify more
interesting distributed algorithms.

In order to be suitable for our purposes, an automaton model must
\begin{itemize}[topsep=1ex,itemsep=0ex]
\item be effectively closed under Boolean set operations,
\item have a decidable emptiness problem, and
\item be able to simulate a single synchronous round of any algorithm
  that can be specified in the corresponding formal language.
\end{itemize}
Hence, NDGAs and ADGAs cannot be used to extend this method (the
former not being closed under complementation, the latter having an
undecidable emptiness problem).

Besides covering a larger class of algorithms, a more expressive
automaton model might also allow us to specify to-be-verified
algorithms in a more natural way, less dependent on the global
controller. With DDGAs, the controller has to compensate for the fact
that we cannot, in general, provide DDGA-recognizable assertions on
the nodes' knowledge about properties of the graph (such as the
diameter in the \textsc{FloodMax} algorithm). If, on the other hand,
we were able to express such assertions, the role of the controller
could be reduced to simply checking whether all the nodes have
terminated. Note that this would inevitably require an automaton model
over infinite node alphabets, since the nodes would have to store
information of unbounded size.

Our verification method also has a limitation that cannot be overcome
by simply switching to another automaton model: it is only applicable
to \emph{synchronous} distributed algorithms. However, this might not
be an issue in practice because any synchronous algorithm can be
(automatically) converted into an asynchronous one using a
synchronizer, as suggested by Awerbuch in \cite{Awe85}. Thus, assuming
the tool used for conversion is correct, a mechanical verification
technique for synchronous algorithms also provides an indirect way of
obtaining verified asynchronous algorithms. Since it is usually easier
to design algorithms for synchronous systems, this seems like a
practical approach.

\subsection{Related Work}
The method presented here is based on the simple observation that we
can reason about a synchronous distributed algorithm as if it were a
sequential algorithm whose elementary operations modify the global
state of an entire graph. This approach has also been recently
employed by Drăgoi et al.\ in \cite{DHV14}, where they have considered
fault-tolerant consensus algorithms operating in a synchronous setting
that allows the topology of the communication graph to change
nondeterministically in every round. (The \textsc{FloodMax} algorithm
from \cref{ex:floodmax} is a simple consensus algorithm.) In order to
verify such algorithms, they have introduced a many-sorted,
first-order-like logic with a very restricted syntax, in which they
can reason about the global state of a graph and its underlying
topology, as well as encode transitions between global states. For
many consensus algorithms, that logic permits to formalize statements
of the following form:
\begin{quote}
  “If state $G_λ$ satisfies invariant $\mathit{Inv}$ and some
  condition on the topology of the graph, and additionally the
  algorithm permits a transition from $G_λ$ to $G'_{λ'}$, then the
  invariant $\mathit{Inv}$ also holds in $G'_{λ'}$.”
\end{quote}
Here, $G_λ$ and $G'_{λ'}$ have the same nodes, but possibly different
edges. In this manner, Drăgoi et al.\ were able to formalize safety
properties and termination of many consensus algorithms from the
literature. The second part of their paper provides a semi-decision
procedure for checking the validity of such verification conditions
expressed in their logic. Furthermore, they could also identify a
decidable fragment of that logic, which, for some algorithms, is
already sufficient to prove correctness.